\RequirePackage{fix-cm}
\documentclass[smallextended]{svjour3}       
\smartqed  
\usepackage{graphicx}
\usepackage{verbatim}
\usepackage{algorithm}
\usepackage{bm}
\usepackage{algpseudocode}
\usepackage[all]{xy}
\usepackage{url}
\usepackage{amsmath,amssymb}
\usepackage{graphicx}
\usepackage{mathtools}

\newtheorem{assumption}{Assumption}

\newcommand{\opt}[1]{{#1}^\ast}

\DeclareMathOperator*{\mini}{\text{Minimize}}

\DeclareMathOperator*{\argmin}{{ \sf{argmin}}}

\newcommand{\powb}[2]{{\left({#1}\right)}^{#2}}

\newcommand{\ic}{\mathbf{j}}

\newcommand{\br}[1]{\left({#1}\right)}

\newcommand{\R}{\mathbb{R}}

\newcommand{\Rep}[1]{\mathrm{Re}\br{#1}}
\newcommand{\Imp}[1]{\mathrm{Im}\br{#1}}
\newcommand{\herm}[1]{\overline{#1}}

\newcommand{\Com}{\mathbb{C}}
\newcommand{\Ical}{\mathcal{I}}

\newcommand{\Desc}[1]{\mathrm{Desc}\br{#1}}
\newcommand{\fp}{S}
\newcommand{\fpp}{P}
\newcommand{\fpq}{Q}
\newcommand{\zr}{r}
\newcommand{\zi}{x}
\newcommand{\yp}{y}
\newcommand{\Inj}{\mathbf{S}}

\newcommand{\cur}{\mathfrak{i}}
\newcommand{\Int}{\mathbb{I}}

\DeclareMathOperator*{\Ext}{\text{Extremize}}
\newcommand{\lb}[1]{{{#1}}^{L}}
\newcommand{\ub}[1]{{{#1}}^{U}}

\newcommand{\vs}{\mathfrak{v}}
\newcommand{\xs}{\beta}

\newcommand{\Par}[1]{\mathcal{P}\br{#1}}
\newcommand{\Chi}[1]{\mathcal{C}\br{#1}}

\newcommand{\Xcal}{\mathcal{X}}

\newcommand{\injp}{p}
\newcommand{\injq}{q}

\newcommand{\inj}{s}

\usepackage{color}
\usepackage{caption}
\usepackage{subcaption}

\usepackage{amssymb}
\usepackage{graphicx}
\usepackage{subfig}
\usepackage{nomencl}
\usepackage{glossaries}

\usepackage{url}
\begin{document}

\title{Graphical models for optimal power flow
}


\author{Krishnamurthy Dvijotham         \and
		Michael Chertkov \and 
        Pascal Van Hentenryck \and 
        Marc Vuffray \and
        Sidhant Misra \thanks{	
	This work was supported by Skoltech through collaboration agreement 1075-MRA. The work at LANL was carried out under the auspices of the National Nuclear Security Administration of the U.S. Department of Energy under Contract No. DE-AC52-06NA25396.}
}


\institute{Krishnamurthy Dvijotham \at
              Computing and Mathematical Sciences, California Institute of Technology\\
              \email{dvij@cs.washington.edu}           
           \and
           Michael Chertkov, Marc Vuffray Sidhant Misra \at
              T-Divison \& Center for Nonlinear Studies, 
              Los Alamos National Laboratory
            \and 
            Pascal Van Hentenryck \\
            Industrial and Operations Engineering, University of Michigan, Ann Arbor  
}

\date{Received: date / Accepted: date}

\maketitle

\begin{abstract}
Optimal power flow (OPF) is the central optimization problem in electric power grids. Although solved routinely in the course of power grid operations, it is known to be strongly NP-hard in general, and weakly NP-hard over tree networks. In this paper, we formulate the optimal power flow problem over tree networks as an inference problem over a tree-structured graphical model where the nodal variables are low-dimensional vectors. We adapt the standard dynamic programming algorithm for inference over a tree-structured graphical model to the OPF problem. Combining this with an interval discretization of the nodal variables, we develop an approximation algorithm for the OPF problem. Further, we  use techniques from constraint programming (CP) to perform interval computations and adaptive bound propagation to obtain practically efficient algorithms. Compared to previous algorithms that solve OPF with optimality guarantees using convex relaxations, our approach is able to work for arbitrary distribution networks and handle mixed-integer optimization problems. Further, it can be implemented in a distributed message-passing fashion that is scalable and is suitable for ``smart grid'' applications like control of distributed energy resources. We evaluate our technique numerically on several benchmark networks and show that practical OPF problems can be solved effectively using this approach.

\end{abstract}

\section{Introduction}

In this paper, we study a novel application of well-known AI
techniques (constraint programming and inference in graphical models)
to a difficult engineering problem - the optimization of resources in
a power distribution network. As larger amounts of renewable
generation sources (solar/wind) are incorporated into the power grid,
the availability of generation capacity becomes uncertain (due to
dependence on unpredictable weather phenomena). However, the physics
of the power grid imply the need to maintain real-time balance between
demand and generation. Exploiting the flexibility of electricity
demand becomes central for solving this problem, which is the core of
the ``smart grid'' vision ({\tt https://www.smartgrid.gov/}). Thus,
future grids will require efficient algorithms that can process data
from millions of consumers and efficiently compute optimal ways of
exploiting demand-side flexibility while respecting engineering
constraints. In this paper, we develop a novel algorithm guaranteed to
compute an approximately optimal solution for this problem in
polynomial time.

More concretely, we study the Optimal Power Flow (OPF) problem
\cite{carpentier1962contribution}. At an abstract level, one can view
this as a network flow optimization:

\begin{align*}
	\text{ minimize } & \text{ cost of generating electicity } \\
	\text{ subject to } & \text{ conservation of flows, flows consistent with voltages,} \\
						& \text{ demands are met}, \text{ engineering limits are respected}
\end{align*}
where the engineering limits typically refer to capacities of the
transmission lines in the power grid and limits on voltages.

However, as opposed to a standard network-flow problem for which there
are well-known efficient algorithms, the physics of electric power
flow make the above problem challenging. Electrical flows cannot be
arbitrary, but are driven by differences in voltages, so that the flow
on a transmission line is a nonlinear function of the voltage difference
between the two ends of the line. Due to this nonlinear constraint,
the OPF problem becomes non-convex. In fact, it is strongly
NP-hard over arbitrary networks \cite{bienstock2015strong} and weakly
NP-hard over tree-structured networks \cite{PascalHardness}. The special
case of tree-structured networks is particularly important in the
context of the smart grid as distribution networks (which connect
high-voltage long-distance power transmission network to individual
consumers) are typically tree-structured. In order to exploit
demand side flexibility, we will need efficient OPF algorithms on
tree-networks. 

In recent years, several researchers have studied applications of
convex relaxation techniques to the OPF problem
\cite{low2014convex}. In particular, for the tree OPF problem, elegant
results have been developed characterizing the conditions under which
convex relaxations of OPF are guaranteed to be exact
\cite{BaseQCQP,lavaei2013geometry,sojoudi2014exactness,gan2015exact}. The
most general results are presented in \cite{gan2015exact} and cover
several practical instances of OPF over tree networks. However, the
conditions for exactness require assumptions that are incompatible
with the above ``smart grid'' applications (absence of discrete
variables, limits on flow reversal, ...).

In this paper, we develop a new approach to solving optimal power flow
on tree-structured networks by using techniques from Constraint
Programming (CP) and Graphical Models (GM). We first restate the tree OPF
problem as an inference problem over a tree-structured factor
graph. Based on this representation, we develop an algorithm that
computes a super-optimal approximately feasible solution. The running time of the algorithm is linear
in the size of the network and polynomial in $\frac{1}{\epsilon}$,
where $\epsilon$ is the error tolerance allowed. Relative to the
existing algorithms based on convex relaxations, the approach we
develop has the following advantages:
\begin{itemize}
\item It can handle mixed-integer optimization problems (involving
  both discrete and continuous variables) and hence is capable of
  addressing load-control and distributed generation applications with
  discrete components such as on/off constraints, switching
  transformer taps, and capacitor banks.

\item Unlike \cite{gan2015exact}, the approach does not require
  restrictive assumptions on flow directionality or voltage limits. It
  also do not require costs to be convex, allowing arbitrary (possibly
  discontinuous) piecewise-linear costs.

\item The resulting algorithm is inherently distributed and can be
  implemented using a message-passing framework, which is expected to
  be a significant advantage for future distribution networks.
\end{itemize}

\noindent
On the other hand, a disadvantage of our algorithm is that we only
produce approximate solutions and there may be cases where achieving
acceptable error tolerances will require intractably fine
discretizations. However, we show that, for several practical OPF
problems, this problem can be alleviated by leveraging CP techniques.


A closely related approach, developed in an abstract form and stated in the language of linear programming, was presented in \cite{bienstock2015lp}. The authors present a general framework that computes super-optimal and approximately feasible solutions to graph-structured mixed-integer polynomial optimization problems. While our approach provides similar guarantees, our main contributions relative to that work are as follows:
\begin{itemize}
	\item[1] We restate the approach in the intuitive language of graphical models. This allows us to take advantage of the inference techniques developed for graphical models \cite{wainwright2008graphical,sontag2010approximate}, and study problems beyond optimization (probabilistic inference, for example).
\item[2] As opposed to the point discretization approach employed in \cite{bienstock2015lp}, we develop an interval discretization approach, that allows us to use interval CP techniques (bound-tightening, constraint propagation etc.) to achieve practically efficient implementation of the algorithm.
\end{itemize}

\noindent
We note that a related approach has been used for finding approximate Nash-equilibria in tree-structured graphical games \cite{kearns2001graphical}. Finally, we note that related ideas have been explored in detail in the graphical models and constraint programming literature \cite{dechter2003constraint,wainwright2008graphical,johnson2008convex,sontag2010approximate}.

The rest of this paper in organized as follows. Section
\ref{sec:TechIntro} provides the background on power systems and
graphical models necessary for this paper. Section \ref{sec:OPFasGM}
formulates the OPF problem over a tree network as a MAP-inference
problem over a tree-structured factor graph. Section \ref{sec:DPInt}
describes a finite dynamic-programming algorithm based on an interval
discretization of the variables in the factor graph and presents
guarantees associated with the algorithm. Section \ref{sec:Num}
evaluates our algorithm and compares it to off-the-shelf solvers on a
number of IEEE distribution network test cases. Finally, Section
\ref{sec:Conc} summarizes our findings and presents directions for
future work.

\section{Background}\label{sec:TechIntro}

In this Section, we introduce all the necessary background on
Alternating Current (AC) power flows and graphical models required to
follow the development of the algorithm in this paper. We have attempted 
to make this Section self-contained providing sufficient details for
the purposes of this paper. Interested readers may consult 
textbooks on power engineering \cite{BergenVittal,pai2014computer} and
graphical models
\cite{koller2009probabilistic,dechter2003constraint,wainwright2008graphical} for further details.

In the following, $\Com$ denotes the set of complex numbers and $\R$
the set of real numbers.$\ic=\sqrt{-1}$ to avoid confusion
with currents as is traditional in power systems. For $x\in\Com$, we
use $\herm{x}$ to denote the complex conjugate of $x$, $\Rep{x}$ the real part,
$\Imp{x}$ the imaginary part, and $\angle x$ the phase of $x$. For $x,y\in\Com$, $a\leq b$ denotes the pair of
inequalities $\Rep{a}\leq \Rep{b},\Imp{a}\leq \Imp{b}$. 

\subsection{AC Power Flow over a Tree Network}


We will work with a power distribution network transporting
Alternating-Current power (AC power). Mimicking power engineering
terminology, nodes in the network are called buses and edges are
called transmission lines (or simply lines or branches). These
networks are typically tree-structured, and the root of the tree is
known as the \emph{substation bus} - physically, this is the point at
which the power distribution network is connected to the high voltage
power transmission network. We label the nodes $0,\ldots,n$, where $0$
is the substation bus. The network is represented as a directed graph,
with all edges pointing towards the substation bus (this
directionality is simply a convention and has no physical
meaning - the physical power flow over this edge can be in either direction). Each node $k$ (except the substation) has a unique outgoing
edge, connecting it to its \emph{parent node}, denoted $\Par{k}$, and
$k$ is said to be a \emph{child} of $\Par{k}$.  $\Chi{i}$ denotes the
set of children of bus $i$. 
\begin{figure}
\centering\includegraphics[width=.95\columnwidth]{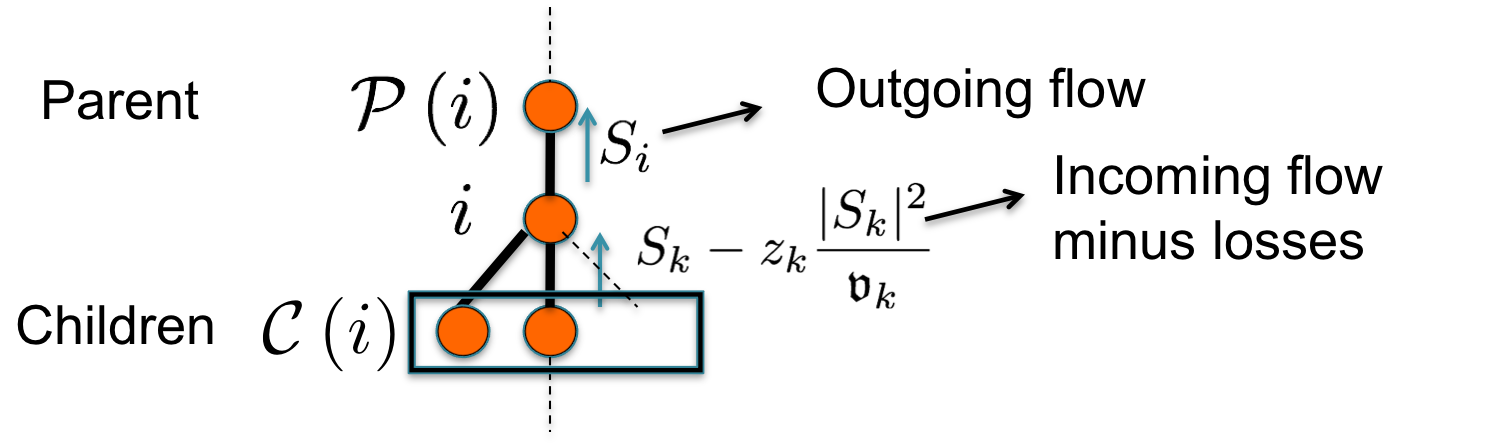}
\caption{AC Power flow in a tree network	}\label{fig:FlowPicB}
\end{figure}  
AC power flow in steady-state is described
by complex voltages (known as voltage phasors) that represent
sinusoidal steady-state voltage profiles. Let the complex voltage at
node $i$ be $V_i$ and let $\vs_i=|V_i|^2$ denote the squared voltage
magnitude. The power-flow is also a complex number, whose real part is
known as \emph{active power} and imaginary part as \emph{reactive
  power}.   In a tree network, every edge connects a bus and its
parent. Thus, we can identify every edge with the child node $i$
incident upon it, and denote its complex impedance by $z_i=r_i+\ic
x_i$ ($r_i$ is the resistance and $x_i$ the inductance of the
line). We denote the sending-end power flow from bus $i$ to bus
$\Par{i}$ by $\fp_i=\fpp_i+\ic \fpq_i$. Note that because of losses,
the power received at $\Par{i}$ is not equal to $\fp_i$. The power
losses are given by $z_{i}\frac{|\fp_i|^2}{\vs_i}$ (see
\cite{BergenVittal} or \cite{low2014convexb} for further details).
Combined with the conservation of flow at each node in the network,
this leads to the AC power flow equations, in the so-called
branch-flow model (first presented in \cite{baran1989optimal}),
illustrated in Figure \ref{fig:FlowPicB}. The power flow equations in
the branch flow (or Baran-Wu) form can be written as:

\begin{subequations}\begin{align}
\fpp_i &=\injp_i+\sum_{k \in \Chi{i}}\br{\fpp_k-r_k\br{\frac{\fpp_k^2+\fpq_k^2}{\vs_k}}}\quad \forall i \in \{0,\ldots,n\}\label{eq:ACpb} \\
\fpq_i &=\injq_i+\sum_{k \in \Chi{i}}\br{\fpq_k-x_k\br{\frac{\fpp_k^2+\fpq_k^2}{\vs_k}}}\quad \forall i \in \{0,\ldots,n\}\label{eq:ACpbreac}  \\
\vs_i &= \vs_k+\br{\zr_k^2+\zi_k^2}\frac{\br{\fpp_k^2+\fpq_k^2}}{\vs_k}-2\br{\zr_i\fpp_i+\zi_i\fpq_i} \quad \forall i \in \{0,\ldots,n\},k \in \Chi{i}\label{eq:ACPFohm}
\end{align}	\label{eq:PF}	
\end{subequations}

\noindent 
where $p_i$ and $q_i$ are the real and reaction power
injections/consumptions at bus $i$, as discussed in the next
Section. For the rest of this paper, this is the form of the PF
equations we will use.

\subsection{Optimal Power Flow (OPF) on a tree network}\label{sec:OPFIntro}

The optimal power flow (OPF) problem aims at finding the most
efficient utilization of generation and flexible demand resources in a
grid subject to the power flow constraints and engineering limits on
voltages, currents and flows. At each node, there is generation and/or
consumption of power. To simplify notation, we assume that there is
only one entity (generator or consumer) at each node in the network (but this restriction is not necessary). Both generators and
flexible consumers are characterized by their injection domain
$\br{\injp_i,\injq_i}\in\Inj_i$. The domain may be a finite set (for modeling
discrete load control where a load can take on one of a set of
possible values, for example) or an infinite set (for modeling
continuous constraints like minimum/maximum generation
limits). Inflexible generators or consumers are modeled by choosing
$\Inj_i$ to be a singleton set. Additionally, each generator has a
cost of production $c_i\br{\injp_i,\injq_i}$ and similarly every
flexible consumer may be compensated for adjusting consumption, and
this compensation is also denoted $c_i\br{\injp_i,\injq_i}$.

With these assumptions, a generic OPF problem can be stated as
\begin{subequations}
\begin{align}
\mini_{\injp,\injq\vs,\fpp,\fpq} &\sum_{i=0}^n c_i\br{\injp_i,\injq_i} \text{ (Minimize Cost)}	\\
\;\;\;\:\;\;\text{ Subject to } 
& \eqref{eq:ACpb},\eqref{eq:ACpbreac},\eqref{eq:ACPFohm} \\
& \vs^{L}_i \leq \vs_i \leq \vs^{U}_i \quad \forall i \in \{1,\ldots,n\} \label{eq:ACPFv} \\
& \lb{\fpp_i} \leq \fpp_i \leq \ub{\fpp_i},\lb{\fpq_i} \leq \fpq_i \leq \ub{\fpq_i} \, \forall i \in \{1,\ldots,n\} \label{eq:ACPFf} \\
& \br{\injp_i,\injq_i}\in \Inj_i\quad \forall i \in \{1,\ldots,n\} \label{eq:ACPFpq} \\
& \fpp_0=0,\fpq_0=0,\vs_0=\vs^{ref} \label{eq:Slackbus} 
\end{align}\label{eq:OPFform}	
\end{subequations}

\noindent
The costs and flow balance constraints are described above. The
constraints \eqref{eq:ACPFv} and \eqref{eq:ACPFf} come from
engineering limits - devices connected to the grid only work properly
for a certain range of voltages, and the flow limits are related to
the capacity of transmission lines (both due to dynamic stability and
thermal limitations). The constraint \eqref{eq:Slackbus} enforces that
there is no current or flow going upstream from the substation bus and
that the substation bus voltage is set to a fixed reference value,
$\vs^{ref}$. Note that more general OPF problems can be handled with
our approach (adding tap transformer positions, capacitor banks etc.),
but we restrict ourselves to this basic problem to simplify notation
and make the exposition clear. We make assumptions on the cost
function and problem data as stated below:
\begin{assumption}\label{assump:A}
The injection constraint set can be partitioned as: 
\begin{align*}
\Inj_i=\cup_{t=1}^{|\Inj_i|}\Int^{\inj}\br{t},\Int^{\inj}\br{t}=\left\{\br{\injp_i,\injq_i}: \lb{\injp_i}\br{t}\leq \injp_i \leq \ub{\injp_i}\br{t}, \lb{\injq_i}\br{t}\leq \injq_i \leq \ub{\injq_i}\br{t}\right\}	
\end{align*}
and the cost function $c_i\br{\injp_i,\injq_i}$ is a linear function over $\Int^{\inj}\br{t}$:
\[c_i\br{\injp_i,\injq_i}=a_i\br{t}\injp_i+b_i\br{t}\injq_i+c_i\br{t}\]
\end{assumption}
\begin{assumption}\label{assump:B}
$\exists M>0$ such that:
\begin{itemize}
	\item[1]  $[\lb{\fpp_i},\ub{\fpp_i}],[\lb{\fpq_i},\ub{\fpq_i}],[\lb{\vs_i},\ub{\vs_i}] \subseteq [-M,M] \quad i\in \{1,\ldots,n\}$ (voltages/flows are bounded uniformly across all nodes).
	\item[3] $\lb{\vs_i}\geq \frac{1}{M}$ for $i=1,\ldots,n$ (voltage lower bounds are bounded below).
	\item[3] $|z_i|\leq M,i=1,\ldots,n$ (impedances are bounded uniformly across all nodes).
	\item[4] $|\Inj_i|\leq M,i=0,\ldots,n$ (number of pieces in the cost is bounded).
	\item[5] $\forall i\in \{0,\ldots,n\},t\in\{1,\ldots,|\Inj_i|\}$: $\max\br{\ub{\injp_i}\br{t}-\lb{\injp_i}\br{t},\ub{\injq_i}\br{t}-\lb{\injq_i}\br{t}} \leq \frac{1}{M}$ (the size of each piece in the piecewise linear cost is small).
\end{itemize}
\end{assumption}

\noindent
Assumptions \ref{assump:A} and \ref{assump:B} are non-restrictive: Assumption \ref{assump:A} simply requires that
the cost function is piecewise-linear (or can be approximated by one) and assumption \ref{assump:B}  requires that all parameters are bounded.

\subsection{Factor Graphs}
\begin{figure}
\centering\includegraphics[width=.23\columnwidth]{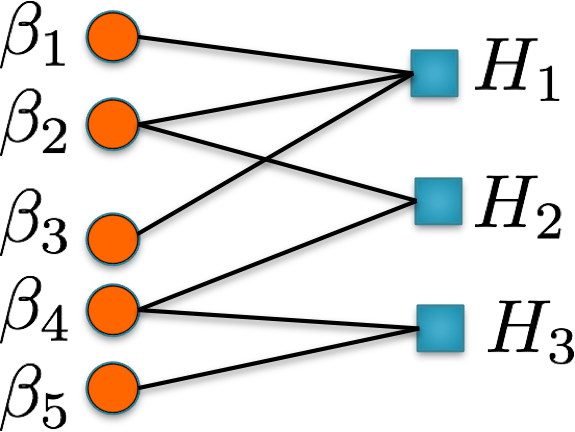}
	\caption{Factor graph corresponding to \eqref{eq:FactExample}}\label{fig:FactExample}	
\end{figure}

A factor graph \cite{kschischang2001factor} is a formal tool to express the structure of an optimization problem or a probability distribution. In this paper, we focus on optimization problems and do not discuss the use of factor graphs in probabilistic inference. A factor graph is defined by specifying:

\begin{itemize}
	\item[1] A set of $n$ variables $\{\xs_i\}_{i=1}^{n}$ and their domains $\xs_i \in \Xcal_i$.
	\item[2] A set of functions, called \emph{factors}, $\{H_{k}:\displaystyle\prod_{t \in \alpha\br{k}} \Xcal_t \mapsto \R\}_{k=1}^m$, where $\alpha\br{k}\subseteq \{1,\ldots,n\}$ is the set of variables that $H_k$ depends on.
\end{itemize}

\noindent
The factor graph is represented as a bipartite graph where the variables live on one side of the graph and the factors on the other side. A variable is connected to a factor if the factor depends on the variable, so that there is an edge between $\xs_i$ and $H_j$ if and only if $i \in \alpha\br{j}$. 
The optimization problem associated with the factor graph is $
\mini_{\{\xs_i \in \Xcal_i\}_{i=1}^n} \sum_{k=1}^m H_k\br{\xs_{\alpha\br{k}}}$. 
As a concrete example, the optimization problem represented by the factor graph shown in Figure \ref{fig:FactExample} is:
\begin{align}
\min_{\{\xs_i \in \{0,1\}\}_{i=1}^5} H_1\br{\xs_1,\xs_2,\xs_3}+H_2\br{\xs_2,\xs_4}+H_3\br{\xs_4,\xs_5}	\label{eq:FactExample}
\end{align}


\noindent
A naive approach to solving \eqref{eq:FactExample} would be to simply enumerate all $2^5$ assignments to the variables $\xs_1,\ldots,\xs_5$. However, one can solve this problem efficiently by exploiting the factor graph structure as follows: We first note that $\xs_5$ only appears in the factor $H_3$, so we can rewrite the optimization as 
\[\min_{\xs_1,\xs_2,\xs_3,\xs_4} H_1\br{\xs_1,\xs_2,\xs_3}+H_2\br{\xs_2,\xs_4}+\min_{\xs_5}H_3\br{\xs_4,\xs_5}\]
Define $\kappa_3\br{\xs_4}=\min_{\xs_5} H_3\br{\xs_4,\xs_5}$ (this function can be evaluated using $4$ units of time assuming each function evaluation takes $1$ unit of time, by evaluating $H_3$ for all $4$ assignments to its arguments). The problem then reduces to
\[\min_{\xs_1,\xs_2,\xs_3} H_1\br{\xs_1,\xs_2,\xs_3}+\min_{\xs_4} H_2\br{\xs_2,\xs_4}+\kappa_3\br{\xs_4}\]
Again, define $\kappa_2\br{\xs_2}=\min_{\xs_4} H_2\br{\xs_2,\xs_4}+\kappa_3\br{\xs_4}$ (this can again be evaluated in $4$ units of time). Then, the problem reduces to 
\[\min_{\xs_1,\xs_2,\xs_3} H_1\br{\xs_1,\xs_2,\xs_3}+\kappa_3\br{\xs_3}\]
No further simplification is possible since $H_1$ depends on all the remaining variables. The optimal value can be computed now in $8$ units of time (since there are $2^3$ possible assignments to the variables).

Thus the global optimum can be computed using $4+4+8=16$ units of time as opposed to the $32$ units of time taken by the naive brute-force approach. This algorithm generalizes to arbitrary tree-structured factor graphs (factor graphs with no cycles). In Section \ref{sec:OPFasGM}, we formulate the ACOPF \eqref{eq:OPFform} as a tree-structured factor graph and show how to exploit factor graph techniques to solve the ACOPF with approximation guarantees.

\section{OPF as a Graphical Model}\label{sec:OPFasGM}

This section shows how to rewrite the OPF problem \eqref{eq:OPFform} as
a graphical model inference problem. We first define the augmented
variables $\xs_i=\begin{pmatrix}
	\vs_i & \fpp_i & \fpq_i \end{pmatrix},i=1,\ldots,n,\xs_0 =\begin{pmatrix}
	\vs^{ref} & 0  & 0
\end{pmatrix}$. Note that $\xs_0$ is fixed  and is only introduced for notational convenience. In terms of the augmented variables $\xs_i$, the constraints \eqref{eq:ACPFv},\eqref{eq:ACPFf} are simply bound constraints on components of $\xs_i$.  The domain of the variable $\xs_i$ is defined as $\Xcal_i =[\lb{\vs}_i,\ub{\vs}_i] \times [\lb{\fpp}_i,\ub{\fpp}_i]  \times [\lb{\fpq}_i,\ub{\fpq}_i]$. Further, using \eqref{eq:ACpb}, we can define
\begin{subequations}
\begin{align}
\injp\br{\xs_i,\xs_{\Chi{i}}} & =\fpp_i-\sum_{k\in\Chi{i}} \br{\fpp_k-\zr_k	\frac{\fpp_k^2+\fpq_k^2}{\vs_k}} \\
\injq\br{\xs_i,\xs_{\Chi{i}}} & =\fpq_i-\sum_{k\in\Chi{i}} \br{\fpq_k-\zi_k	\frac{\fpp_k^2+\fpq_k^2}{\vs_k}} \\
\vs_i\br{\xs_k} & = \vs_k+\br{\zr_k^2+\zi_k^2}\frac{\br{\fpp_k^2+\fpq_k^2}}{\vs_k}-2\br{\fpp_k\zr_k+\fpq_k\zi_k}
\end{align}	
\end{subequations}

\noindent
For each $i=0,\ldots,n$, we define 
\begin{align}
&H_i\br{\xs_i,\xs_{\Chi{i}}}= \\
&\begin{cases}
c_i\br{\inj_i}  & \text{ if } 	\left\{\begin{array}{l} \br{\injp_i\br{\xs_i,\xs_{\Chi{i}}},\injq_i\br{\xs_i,\xs_{\Chi{i}}}}  \in \Inj_i	\\
 \vs_i  = \vs_i\br{\xs_k} \quad \forall k \in \Chi{i} \end{array}\right. \\
\infty & \text{ otherwise } 
\end{cases}	\label{eq:Hdef}
\end{align}
This is a convenient shorthand that allows to write the constrained optimization problem \eqref{eq:OPFform} as an unconstrained problem, thereby simplifying notation. The OPF problem \eqref{eq:OPFform} is equivalent to
\begin{align}
\mini_{\{\xs_i\in \Xcal_i\}_{i=1}^n} \sum_{i=0}^n H_i\br{\xs_i,\xs_{\Chi{i}}}\label{eq:optDP}	
\end{align}
This corresponds to a graphical model in the factor graph representation \cite{kschischang2001factor} where the nodal variables are $\xs_i$ and the factors correspond to the $H_i$. 
\begin{definition}[OPF factor graph]\label{def:FactG}
The problem \eqref{eq:optDP} corresponds to a factor graph: The set of variable nodes is 
\[\xs_i\in \Xcal_i \;\;\; (i=1,\ldots,n)\] and the set of factors is 
\[H_i \;\;\; (i=0,\ldots,n)\]
For $i=1,\ldots,n$, $\xs_i$ is connected to $H_i,H_{\Par{i}}$ (see Figure \ref{fig:FactGraph}).
\end{definition}

\begin{theorem}
The factor graph from definition \ref{def:FactG} is a tree. Hence, the problem \eqref{eq:optDP} can be solved exactly by a two-pass dynamic programming algorithm.
\end{theorem}
\begin{proof}
Every variable node $\xs_i \;\; (i=1,\ldots,n)$ is connected to two factors $H_{\Par{i}},H_i$.  Thus, the total number of edges in the factor graph is $2n$ and the total number of nodes is $2n-1$ (there are factors $H_0,\ldots,H_n$ and variable nodes $\xs_1,\ldots,\xs_n$). Further, the factor graph is connected because variable nodes are connected (Since the variable nodes that are also neighbors in the power network are neighbors 2 hops apart in the factor graph). The factor graph is a connected graph with the number of vertices equal to the number of edges plus one. Therefore, the graph is a tree. The result on efficient inference over a tree factor graph is a standard result \cite{kschischang2001factor}.
\end{proof}
\begin{figure}
\centering\includegraphics[width=.55\textwidth]{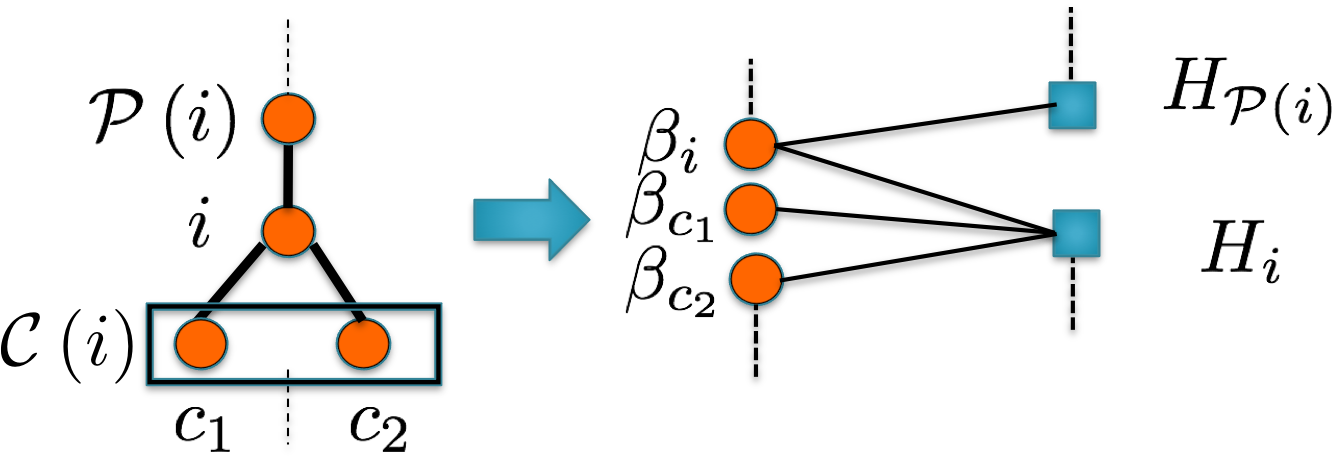}	
\caption{Transformation from the power network to factor graph: Each variable node $i$ is connected to two factors - $H_{\Par{i}}$ and $H_i$}	\label{fig:FactGraph} 
\end{figure}

\noindent
The dynamic-programming (DP) approach is formalized in Algorithm
\ref{Alg:DPMain}. The algorithm works by passing ``messages''
$\kappa_t\br{\xs_t}$. Let $\Desc{t}$ denote all the nodes in the
subtree rooted at $t$ (the descendants of $t$, including $t$ itself).

\begin{algorithm}
\begin{algorithmic}
\State $\mathrm{Processed}=\emptyset$
\State $\kappa_i\br{\xs_i}=0,\eta_i\br{\xs_i}=\emptyset,i=0,1\ldots,n$
\While{$|\mathrm{Processed}|<=n$} 
\State Choose $i \in \{k: k \not\in \mathrm{Processed},\Chi{i} \subseteq \mathrm{Processed}\}$
\State  
\begin{subequations}
\begin{align}
	\kappa_{i}\br{\xs_i} \gets \displaystyle\min_{\xs_{\Chi{i}}}H_i\br{\xs_i,\xs_{\Chi{i}}} +\sum_{k \in \Chi{i}}\kappa_{k}\br{\xs_k}  	\\
	\eta_i\br{\xs_i} \gets \displaystyle\argmin_{\xs_{\Chi{i}}}H_i\br{\xs_i,\xs_{\Chi{i}}} +\sum_{k \in \Chi{i}}\kappa_{k}\br{\xs_k}
  \end{align}	\label{eq:DPMain}
\end{subequations}
\State $\mathrm{Processed}=\mathrm{Processed}\cup\{i\}$
\EndWhile
\State $\br{\opt{c},\opt{\xs_{\Chi{0}}}} \gets\displaystyle\min_{\xs_{\Chi{0}}}H_0\br{\xs_0,\xs_{\Chi{0}}} +\sum_{k \in \Chi{0}}\kappa_{k}\br{\xs_k} 	$
\State $\mathrm{Processed}=\{0\}$. 
\While{$|\mathrm{Processed}|\leq n$} 
\State Choose $i \in \{k: k \not\in \mathrm{Processed},\Par{k}\in \mathrm{Processed}\} $
\State $\opt{\xs}_{\Chi{i}}\gets\eta_{i}\br{\opt{\xs}_i}$
\State $\mathrm{Processed}\gets \mathrm{Processed} \cup \Chi{i}$
\EndWhile  \\
\Return $\br{\opt{c},\opt{\xs}}$
\end{algorithmic}
\caption{DP Algorithm: (optimization implicitly subject to $\xs_i \in \Xcal_i$)}	\label{Alg:DPMain}
\end{algorithm}

The message denotes the optimal value of the following subproblem of\eqref{eq:optDP}:
\begin{subequations}
\begin{align*}
\kappa_t\br{\xs_t}=& \min_{\xs_i\in \Xcal_i :i \in \Desc{t}\setminus \{t\}} \sum_{i\in \Desc{t}} H_i\br{\xs_i,\xs_{\Chi{i}}}. \end{align*}
\end{subequations}
The messages can be computed recursively starting at the leaves of the tree. 
For the 5 bus network shown in Figure \ref{fig:5buspic}, Algorithm \ref{Alg:DPMain} proceeds as follows: At $T=0$, nodes $3,4$ send messages $\kappa_3,\kappa_4$ (which are simply equal to $H_3,H_4$) to their parent nodes. For brevity, we do not write the constraint $\xs_i \in \Xcal_i$ in the updates below explicitly. At $T=1$, node $2$ computes 
\[\kappa_2\br{\xs_2}=\min_{\xs_3\in \Xcal_3} H_2\br{\xs_2,\xs_3}+\kappa_3\br{\xs_3},\eta_2\br{\xs_2}=\argmin_{\xs_3\in \Xcal_3} H_2\br{\xs_2,\xs_3}+\kappa_3\br{\xs_3}.\]At $T=2$, node $1$ computes 
\begin{align*}
\kappa_1\br{\xs_1}=\min_{\xs_2\in \Xcal_2,\xs_4 \in \Xcal_4} H_2\br{\xs_1,\xs_2,\xs_4}+\kappa_2\br{\xs_2}+\kappa_4\br{\xs_4}	\\
\eta_1\br{\xs_1} = \argmin_{\xs_2\in \Xcal_2,\xs_4 \in \Xcal_4} H_2\br{\xs_1,\xs_2,\xs_4}+\kappa_2\br{\xs_2}+\kappa_4\br{\xs_4}.
\end{align*}
The forward pass ends here, and the backward pass proceeds similarly in the reverse direction:
\begin{align*}
T=3: & \br{\text{OPT},\opt{\xs_1}}=\min_{\xs_1 \in \Xcal_1} \kappa_1\br{\xs_1} \\
T=4: & \br{\opt{\xs_2},\opt{\xs_4}}=\eta_1\br{\opt{\xs_1}}\\
T=5: & \opt{\xs_3}=\eta_2\br{\opt{\xs_2}}.
\end{align*}

\begin{figure}
       \centering \includegraphics[width=.32\textwidth]{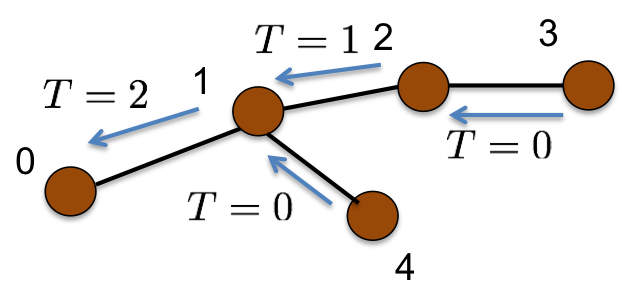}
       \caption{Transformation of ACOPF to Graphical Model} \label{fig:5buspic}
\end{figure}

\section{Finite Algorithm based on Interval Discretization}
\label{sec:DPInt}
The algorithm described in the preceding section, as it stands, cannot
be implemented on a computer, since it requires representing the
functional objects $\kappa^i\br{\xs_i}$.  A straightforward approach to deal
with this would be to discretize the variables $\xs_i$ to a finite set
and allow for some error tolerance on each of the constraints in
\eqref{eq:OPFform}. However, our experiments have indicated that, in
order to produce solutions of acceptable accuracy, one needs an
intractably fine discretization of $\xs_i$ resulting in prohibitive
computation times, and that an accurate estimation of the error
tolerance parameter can be problematic.

Hence, we need an alternative procedure to approximate the infinite dimensional dynamic program. We take the approach of using an interval-based discretization of the power flow variables $\vs_i,\fpp_i,\fpq_i$ (i.e., each variable can take values in the respective interval). Given the constraints $\xs_i \in \Xcal_i$, we partition the set $\Xcal_i$ into a finite union of interval regions defined by interval constraints on the components of $\xs_i$. For any practical OPF problem, $\Xcal_i$ is a compact set (since the bounds on voltages/flows are always finite), so such a decomposition is always possible. Naturally, the computational complexity of the algorithm depends on the number of interval regions. If we fix an interval resolution $\epsilon$ (each interval region in our decomposition is made up of interval constraints of width at most $\epsilon$ in each variable), the number of interval regions depends on the bounds defining the region $\Xcal_i$. Thus, it is of interest to have as tight bounds as possible on the variables $\xs_i$. We describe a procedure to infer tight bounds on the power flow variables using convex relaxations in Section \ref{sec:Relax}. We use the inferred bounds to redefine $\Xcal_i$ and perform the interval discretization on this refined domain.
Then, we perform the dynamic programming update in the space of intervals: For every interval-value of the parent, we look at all possible interval values of the children, and for each combination of intervals, compute the minimum cost solution given the interval constraints. This gives us a lower bound on the message for each interval value of the parent, thus leading to a piecewise-constant lower bound on the message functions $\kappa^i\br{\xs_i}$.  The algorithm is described in Section \ref{sec:DPInt1}.

As the interval discretization gets finer, the relaxation becomes tighter, reducing the errors incurred in the power flow equations due to the relaxation. The errors can be made smaller than $\epsilon$ for any $\epsilon>0$, with the running time polynomial in $\frac{1}{\epsilon}$. These results are formalized in Theorem \ref{thm:ApproxMain} in Section \ref{sec:DPThm}.

\subsection{Interval Dynamic Programming with Adaptive Refinement}\label{sec:DPInt1}

We use Algorithm \ref{Alg:DPMain} with the DP update step replaced with a tractable lower bound based on an interval partition of the variable domains. We develop such a scheme guaranteed to produce a lower bound on the optimal value and an approximately feasible solution. The algorithm is based on a set of operators that will replace the DP update \eqref{eq:DPMain}:
\begin{itemize}
	\item[1]\emph{Interval Partition Operator}: A procedure that take a set and creates a collection of intervals such that every point in the set is in an interval box of size at most $\epsilon$, for some specified tolerance $\epsilon$.
	\item[2] \emph{Interval DP Relaxation}: A procedure that takes a set of interval constraints on children of a node $i$ and produces a lower bound on  $\kappa^i\br{\xs_i}$. 
\end{itemize}
We now define these formally. 

\begin{definition}
An interval constraint on a real variable $x\in\R$ is a constraint of the type $a \leq x \leq b$, parameterized by real numbers $a,b, a\leq b$. 
\end{definition}

\begin{definition}
An interval region $\Int_i$ is a subset of $\Xcal_i$ specified by interval constraints on the variables $\vs_i,\fpp_i,\fpq_i$:
\begin{align*}
\Int_i=\left\{\br{\vs_i,\fpp_i,\fpq_i}:	\begin{array}{ccc}
\lb{\vs_i} \leq \vs_i \leq \ub{\vs_i} \\
\lb{\fpp_i} \leq \fpp_i \leq \ub{\fpp_i} \\
\lb{\fpq_i} \leq \fpq_i \leq \ub{\fpq_i} 
\end{array}\right\}	
\end{align*}
We use  $\vs\br{\Int_i}=[\lb{\vs_i},\ub{\vs_i}],\fpp\br{\Int_i}=[\lb{\fpp_i},\ub{\fpp_i}],\fpq\br{\Int_i}=[\lb{\fpq_i},\ub{\fpq_i}]$ to denote the interval constraints in $\Int_i$ corresponding to each of the variables,
\begin{align*}
& \mathrm{mid}\br{\Int_i}	  = \br{\frac{\lb{\vs_i}+\ub{\vs_i}}{2},\frac{\lb{\fpp_i}+\ub{\fpp_i}}{2},\frac{\lb{\fpq_i}+\ub{\fpq_i}}{2}}
\end{align*}
to select the midpoint of an interval region and
\[\mathrm{Rad}\br{\Int_i}=\max\br{|\lb{\vs_i}-\ub{\vs_i}|,|\lb{\fpp_i}-\ub{\fpp_i}|,|\lb{\fpq_i}-\ub{\fpq_i}|}.\]
\end{definition}

\begin{definition}
An interval partition $\Ical_i$ of a set $X \subseteq \Xcal_i$ is a collection of interval regions satisfying the following conditions:
\begin{subequations}
\begin{align}
& \Ical_i\br{t}  = \left\{\xs_i: \begin{array}{ccc} 
\lb{\vs_i}\br{t}  \leq & \vs_i  & \leq \ub{\vs_i}\br{t} \\
\lb{\fpp_i}\br{t}  \leq & \fpp_i  & \leq \ub{\fpp_i}\br{t} \\
\lb{\fpq_i}\br{t}  \leq & \fpq_i  & \leq \ub{\fpq_i}\br{t} 
 \end{array}\right\},t=1,\ldots,|\Ical_i| \label{eq:IntervalDefa}\\
& \cup_{t=1}^{|\Ical_i|} \Ical_i\br{t}  = X	\label{eq:IntervalDefb}\end{align}\label{eq:IntervalDef}	
\end{subequations}

\noindent
The number of regions in the partition is denoted by $|\Ical_i|$. A partition such that $
\max_t \mathrm{Rad}\br{\Ical_i\br{t}} \leq \epsilon $ 
is denoted as $\mathrm{Partition}\br{\Xcal_i;\epsilon}$.
\end{definition}

\begin{definition}\label{def:IntervalRelax}
An interval relaxation of the DP update step \eqref{eq:DPMain} is a computational procedure that, given a bus $i \in \{0,\ldots,n\}$ with $\Chi{i}=\{k_1,\ldots,k_m\}$ and interval regions $\xs_i \in \Int_i \subseteq \Xcal_i,\xs_k \in \Int_k\subseteq \Xcal_k$ for each $k\in\Chi{i}$, with $\max_{k\in \{i\}\cup\Chi{i}}\mathrm{Rad}\br{\Int_k}\leq \epsilon$, produces as output an interval region $\Int_i^{\prime}$ and values $\injp_i,\injq_i$ such that 
\begin{subequations}
\begin{align}
& \exists \{\xs_k\in \Int_k\}_{k \in \Chi{i}},\xs_i\in \Int_i^{\prime} \nonumber \\
& \text{ s.t } \left\{\begin{array}{ll} |\injp_i-\injp_i\br{\xs_i,\xs_{\Chi{i}}}| &  \leq \eta \epsilon \\
|\injq_i-\injq_i\br{\xs_i,\xs_{\Chi{i}}}| &  \leq \eta \epsilon \\
 \left| \vs_i -\br{\vs_k+\frac{\fpp_k^2+\fpq_k^2}{\vs_k}\br{\zr_k^2+\zi_k^2}-2\br{\fpp_k\zr_k+\fpq_k\zi_k}}\right| & \leq \eta \epsilon \quad \forall k \in \Chi{i} 
 	\end{array}\right. \label{eq:IntervalRelaxDefa}\\
 & c_i\br{\injp_i,\injq_i} \leq \min_{\{\xs_k\in \Int_k\}_{k \in \Chi{i}\cup\{i\}}} H_i\br{\xs_i,\xs_{\Chi{i}}} \label{eq:IntervalRelaxDefb}\\
 & \{\xs_i:\exists \{\xs_k \in \Int_k\}_{k \in \Chi{i}} \text{ s.t }H_i\br{\xs_i,\xs_{\Chi{i}}}<\infty\} \subseteq \Int_i^\prime \subseteq \Int_i \label{eq:IntervalRelaxDefc}
\end{align}\label{eq:IntervalRelaxDef}
\end{subequations}
\noindent
where $\eta$ is a constant that depends only on the number $M$ from Assumption \ref{assump:B}. We denote this computation as 
\[\br{\Int_i^\prime,\injp_i,\injq_i}=\mathrm{PropBound}\br{i,\Int_i,\Int_{k_1},\ldots,\Int_{k_m}}.\]
\end{definition}
These conditions can be interpreted as follows:
\begin{itemize}
	\item  \eqref{eq:IntervalRelaxDefa} states that the relaxation gets tighter as the intervals get smaller. A natural relaxation that satisfies this requirement is to take convex envelopes of the nonlinear terms over the bound constraints defined by the intervals: This is what we use in the concrete implementation described in Section \ref{sec:ConcreteIntervalRelax}. 
	\item \eqref{eq:IntervalRelaxDefb} states that the injections produced by the relaxation step are super-optimal, so that we are guaranteed to get a lower bound on the optimal solution through the DP procedure.
	\item \eqref{eq:IntervalRelaxDefc} states that the bound propagation (which shrinks the interval $\Int_i$ to $\Int_i^{\prime}$ using the constraints implicit in $H_i$) cannot cut off any feasible points.
\end{itemize}

\noindent
Given this computational procedure, we can construct a DP-like
algorithm where the intractable DP update \eqref{eq:DPMain} is
replaced with a tractable procedure based on $\mathrm{PropBound}$,
thereby producing a lower bound on the message function
$\kappa^i\br{\xs_i}$. The algorithm is described in Algorithms
\ref{Alg:DPUpdate} and \ref{Alg:DPUpdateLeaf}. The algorithm starts by
partitioning the space $\Xcal_i$ using the interval partition
operator. For each element of the interval partition, we loop over the
pieces in the messages corresponding to each child, and propagate
constraints from the children to the parent $\xs_i$. If the propagated
interval is non-empty (that is, there exists a feasible setting for
the parents and children within the interval constraints), the lower
bound computed on $c_i\br{\injp_i,\injq_i}$ is used and a new piece is
added to the messages $\kappa^i\br{\xs_i},\eta^i\br{\xs_i}$. In
comparison to the DP Algorithm \ref{Alg:DPMain}, we also maintain
functions $\injp^i,\injq^i,\xs^i$ which store the optimal injections
and intervals for every variable computed in the DP procedure.
%


\subsection{Implementation of operators}

\subsubsection{Interval Discretization Operator}
The $\epsilon$-partition operator can be implemented by using a uniform discretization. The bounds $\lb{\vs},\ub{\vs},\lb{\fp},\ub{\fp}$ are obtained from the bound tightening procedure \eqref{eq:BoundTighten} described in Section \ref{sec:Relax}. For the variable $\vs_i$ with bounds $\lb{\vs_i},\ub{\vs_i}$, the partition operator will create the intervals
\[\left\{[\lb{\vs_i},\lb{\vs_i}+\epsilon^{\vs}],[\lb{\vs_i}+\epsilon^{\vs},\lb{\vs_i}+2\epsilon^{\vs}],\ldots,\left[\lb{\vs_i}+\left\lceil\frac{\ub{\vs_i}-\lb{\vs_i}}{\epsilon^{\vs}}\right\rceil \epsilon^{\vs}\right]\right\}\]
A similar interval discretization procedure is used for $\fpp_i,\fpq_i$. 

\subsubsection{Interval Relaxation Operator}\label{sec:ConcreteIntervalRelax}

In order to define the interval relaxation operator, it will be convenient to introduce the square of the current magnitude $\cur_i=\frac{\fpp_i^2+\fpq_i^2}{\vs_i}$. This serves to isolate the nonconvexities in the problem and simplify the derivation of convex relaxations. Note also that using the current is natural in view of its edge-invariance: in contrast to power flows, current conserves along any edge (power line).

The interval relaxation operator requires solution of the following problem:
\begin{subequations}
\begin{align}
 \Ext_{\{\vs_k,\fpp_k,\fpq_k\}_{k \in \Chi{i}\cup\{i\}}} & \{\vs_i,\fpp_i,\fpq_i,c_i\br{\injp_i,\injq_i}\} \\
 \text{Subject to } & \fpp_i = \injp_i+\sum_{k \in \Chi{i}}\br{\fpp_k-\cur_k\zr_k} \label{eq:PFnewa}\\
& \fpq_i = \injq_i+\sum_{k \in \Chi{i}}\br{\fpq_k-\cur_k\zi_k} \label{eq:PFnewb}\\
& \vs_i = \vs_k+\cur_k\br{\zr_k^2+\zi_k^2}-2\br{\fpp_k\zr_k+\fpq_k\zi_k},k \in \Chi{i} \label{eq:PFnewc}\\
& \vs_k\cur_k= \fpp_k^2+\fpq_k^2, k \in \Chi{i}\cup\{i\} \label{eq:PropNC}\\
& \vs_k \in \vs\br{\Int_k}, \fpp_k\in \fpp\br{\Int_k},\fpq_k\in \fpq\br{\Int_k},k\in\{i\}\cup\Chi{i}  \label{eq:PFnewd}\\
& \br{\injp_i,\injq_i} \in \Inj_i
\end{align}	\label{eq:PropBound}	
\end{subequations}

\noindent
where $\Ext$ means that we both maximize and minimize every term in the set of objectives subject to the constraints specified. Thus, we obtain tighter bounds on the variables $\vs_i,\fpp_i,\fpq_i$ and a lower bound on the objective given the interval constraints. The nonconvexity in the above problem is due to the constraint \eqref{eq:PropNC} and the possibly nonconvex cost $c_i$ and constraints $\Inj_i$. To deal with the later, we explicitly enumerate over $\Int_i^{\inj}\br{t},t=1,\ldots,|\Inj_i|$ so that for each $t$, the injection constraints and costs are linear. To deal with nonconvexity of \eqref{eq:PropNC}, we use convex envelopes of the bilinear and quadratic terms. An abstract version of the problem solved is below (we defer exact details to Section \ref{sec:AppRelax})

\begin{subequations}
$\forall t \in \{1,\ldots,|\Inj_i|\}$
\begin{align}
 \Ext_{\injp_i,\injq_i,\{\vs_k,\fpp_k,\fpq_k\}_{k \in \Chi{i}\cup\{i\}}} & \{\vs_i,\fpp_i,\fpq_i,a_i\br{t}\injp_i+b_i\br{t}\injq_i+c_i\br{t}\} \\
 \text{Subject to } & \eqref{eq:PFnewa},\eqref{eq:PFnewb},\eqref{eq:PFnewc} \\
& 0 \in \text{Relax}\br{\vs_k\cur_k-\br{\fpp_k^2+\fpq_k^2},\Int_k}, k \in \Chi{i}\cup\{i\} \label{eq:PropNC}\\
& \vs_k \in \vs\br{\Int_k}, \fpp_k\in \fpp\br{\Int_k},\fpq_k\in \fpq\br{\Int_k},k\in\{i\}\cup\Chi{i}  \\
	& \injp_i\in [\lb{\injp_i}\br{t},\ub{\injp_i}\br{t}],\injq_i\in [\lb{\injq_i}\br{t},\ub{\injq_i}\br{t}]
\end{align}	\label{eq:PropBoundRelax}	
\end{subequations}

\noindent
The relaxations we use depend on the bound constraints $\Int_k$ and are denoted as $\text{Relax}\br{\vs_k\cur_k-\br{\fpp_k^2+\fpq_k^2},\Int_k}$. A simple example of this kind of relaxation is shown pictorially in Figure \ref{fig:ConvexRelaxInt} for the constraint $xy=1$. It is easy to see that this relaxation gets tighter as the bound constraints on $x$ get tighter, leading to the property that any feasible solution $\br{x,y}$ of the relaxation satisfies $|xy-1| \propto \epsilon$, where $\epsilon$ is the size of the interval constraint on $x$ (this is formalized in Lemma \ref{lem:AppRelax} in the Appendix Section \ref{sec:AppRelax}).
\begin{figure}
\begin{center}
\includegraphics[width=.4\textwidth]{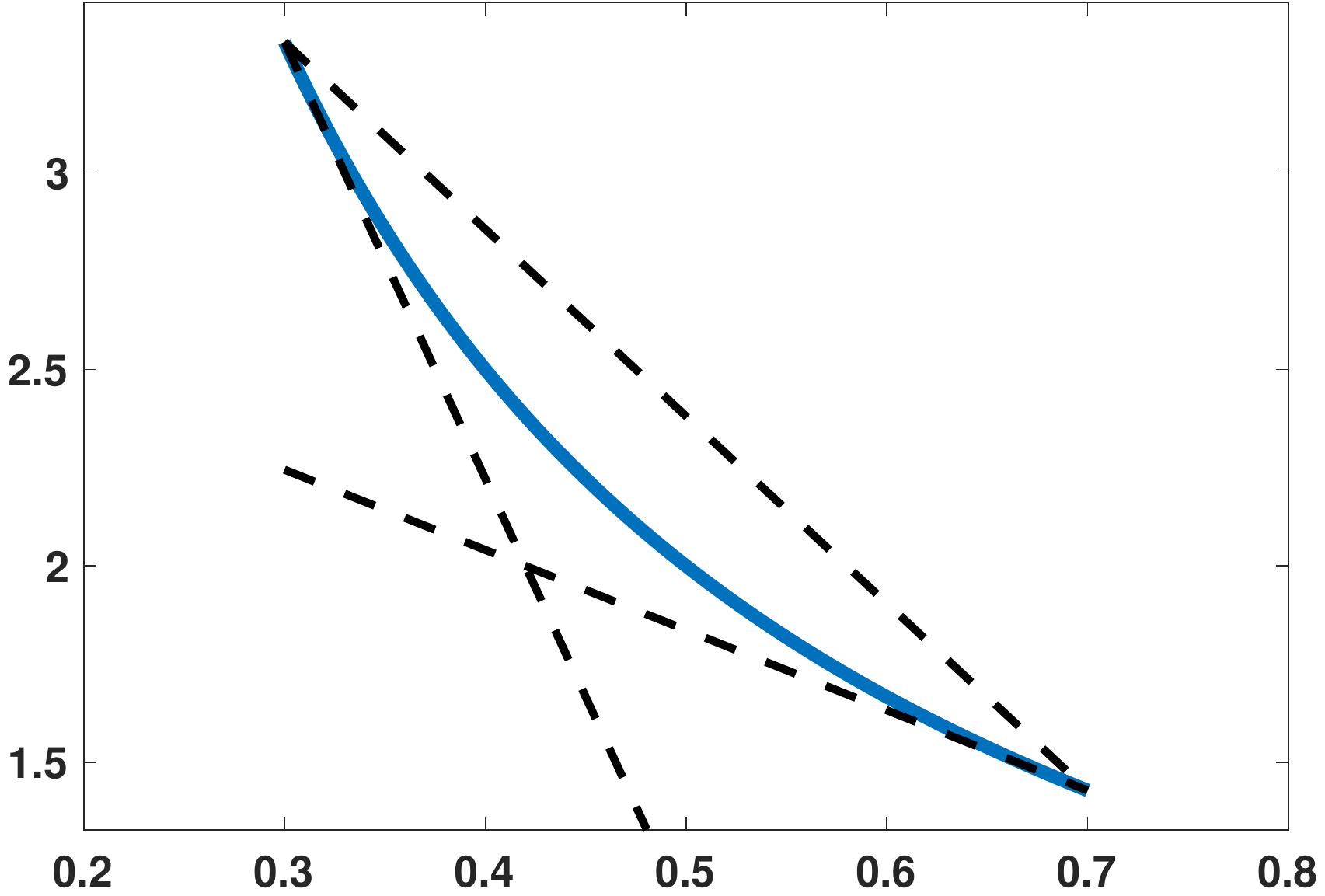}	
\end{center}	
\caption{Convex relaxation of nonlinear constraint $xy=1$ over the region $x\in [0.3,0.7],y\in [\frac{1}{0.7},\frac{1}{0.3}]$:  The set of points satisfying the constraint $xy=1,x \in [.3,.7]$ is plotted in the solid curve (blue). The convex region enclosed by the dashed lines (black) is the feasible region of the convex relaxation. }\label{fig:ConvexRelaxInt}
\end{figure}

\begin{algorithm}
\begin{algorithmic}
\State $\br{j,k} \gets \Chi{i}$
\State $n_i \gets 0$
\State $\Ical_i \gets \mathrm{Partition}\br{\Xcal_i,\epsilon}$
\For{$m_i=1,\ldots,|\Ical_i|$}
\For{$m_j \in 1,\ldots,|\eta^j|$}
\For{$m_k \in 1,\ldots,|\eta^k|$}
\State $\br{\Int,\injp,\injq} \gets \mathrm{PropBound}\br{\Ical_i\br{m_i},\xs^{j}\br{m_j},\xs^k\br{m_k}}$
\If {$\Int \neq \emptyset$}
	\State $n_i \gets n_i+1,\eta^i\br{n_i}\gets \br{m_j,m_k}$
	\State $\xs^i\br{n_i} \gets \Int,\kappa^i\br{n_i} \gets c_i\br{p,q}+\kappa^j\br{m_j}+\kappa^k\br{m_k}$
	\State $p^i\br{n_i} \gets \injp,q^i\br{n_i} \gets \injq$
\EndIf
\EndFor
\EndFor
\EndFor
\State \Return $\eta^i,\kappa^i,\injp^i,\injq^i,\xs^i$
\end{algorithmic}
\caption{Interval DP update at node $i$ with children $\br{j,k}$}	\label{Alg:DPUpdate}
\end{algorithm}

\begin{algorithm}
\begin{algorithmic}
\State $n_i \gets 0$
\For{$m_i \in \left\{1,\ldots,\left\lceil \frac{\ub{\vs_i}-\lb{\vs_i}}{\epsilon}\right\rceil\right\}$}
\For{$t=1,\ldots,|\Inj_i|$}
\State $n_i=n_i+1$
\If {$[\lb{\fpp},\ub{\fpp}]\cap [\lb{\injp_i\br{t}},\ub{\injp_i\br{t}}]\neq \emptyset,[\lb{\fpq},\ub{\fpq}]\cap [\lb{\injq_i\br{t}},\ub{\injq_i\br{t}}]\neq \emptyset$}
\State $\Int^{\vs}  \gets [\lb{\vs_i}+\br{m_i-1}\epsilon,\min\br{\lb{\vs_i}+m_i\epsilon,\ub{\vs_i}}],\Int^{\fpp} \gets  [\lb{\fpp},\ub{\fpp}]\cap [\lb{\injp_i\br{t}},\ub{\injp_i\br{t}}]$
\State $\Int^{\fpq} \gets [\lb{\fpq},\ub{\fpq}]\cap [\lb{\injq_i\br{t}},\ub{\injq_i\br{t}}],\eta^{i}\br{n_i} \gets \Int^{\vs}\times \Int^{\fpp} \times \Int^{\fpq}$
\State $\br{\injp^i\br{n_i},\injq^i\br{n_i}}\gets \displaystyle\argmin_{\injp_i\in \Int^{\fpp},\injq\in \Int^{\fpq}} c_i\br{t}+a_i\br{t}\injp_i+b_i\br{t}\injq_i$
\State $\kappa^i\br{n_i} \gets \displaystyle\min_{\injp_i\in \Int^{\fpp},\injq_i\in \Int^{\fpq}} c_i\br{t}+a_i\br{t}\injp_i+b_i\br{t}\injq_i$
\EndIf
\EndFor
\EndFor
\State \Return $\eta^i,\kappa^i$
\end{algorithmic}
\caption{Interval DP update at leaf node}	\label{Alg:DPUpdateLeaf}
\end{algorithm}

\subsection{Analysis of the interval DP algorithm} \label{sec:DPThm}

We now present formal results verifying the correctness, optimality, and feasibility properties of the solutions produced by our DP algorithm. 

\noindent
Before we state our main theorem that provides an approximation guarantee, we note that we can always convert an OPF problem on an arbitrary tree network to a problem on a tree network with maximum degree $3$:
\begin{lemma}\label{lem:DegLem}
An OPF problem on an arbitrary tree network with $n$ nodes and maximum degree $d$ can be converted to an OPF problem on a modified tree network with maximum degree $3$ and at most $nd$ nodes.
\end{lemma}
\begin{proof}
See Appendix Section \ref{sec:AppDeg}.	
\end{proof}

\begin{theorem}[Approximate optimality property]\label{thm:ApproxMain}
Suppose that assumptions \ref{assump:A} and \ref{assump:B} (see Section \ref{sec:OPFIntro}) hold and that the DP Algorithm \ref{Alg:DPMain} (with update rule from Algorithm \ref{Alg:DPUpdate}) is run on a tree network with maximum degree $3$ and with $0<\epsilon<1$. Let $\opt{m}_1,\opt{m}_2,\ldots,\opt{m}_n$ be the indices of the variables in the optimal solution. 	Let $\br{\opt{\vs}_i,\opt{\fpp}_i,\opt{\fpq}_i}=\mathrm{mid}\br{\eta_i\br{\opt{m}_i}}$ and $\opt{\injp}_i=\injp^i\br{\opt{m}_i},\opt{\injq}_i=\injq^i\br{\opt{m}_i}$. Then, the following guarantees hold:
\begin{itemize}
	\item[1]\emph{Approximation guarantee}: $\br{\opt{\vs},\opt{\fpp},\opt{\fpq},\opt{\injp},\opt{\injq}}$ satisfies  each constraint of \eqref{eq:OPFform} with a bounded error $\zeta \epsilon$ where $\zeta$ is a constant that depends only on $M$ (the constant from Assumption \ref{assump:B} in Section \ref{sec:OPFIntro}).
	\item[2] \emph{Runtime bound}: There is a constant $\zeta^\prime$ (depending on $M$) such that the algorithm requires at most $n  \zeta^\prime \powb{\frac{1}{\epsilon}}{5}$ calls to the $\mathrm{PropBound}$ routine.
	\item[3] \emph{Optimality guarantee}: The cost of the solution is bounded as:
\[\sum_{i=0}^n c_i\br{\opt{p}_i,\opt{q}_i} \leq \mathrm{OPT}\]
where $\mathrm{OPT}$ is the optimal cost of the original problem \eqref{eq:OPFform}.
\end{itemize}
Thus, we find a super-optimal approximately feasible solution in time linear in the size of the network and polynomial in the error tolerance.
\end{theorem}
\begin{proof}
See Appendix Section \ref{sec:AppMain}.
\end{proof}

\begin{remark} Theorem \ref{thm:ApproxMain} formalizes the intuition that as we use finer intervals in the Interval DP algorithm, we get closer to the optimal solution in terms of cost, and we get a tighter relaxation as well. The numerical results in Section \ref{sec:Num} show that our algorithm often finds the true optimal solutions even with a finite error tolerance.
\end{remark}


\section{Numerical Illustrations}\label{sec:Num}

In this Section, we present numerical tests of our approach on some IEEE benchmark networks - power grids with network topologies and loads that are deemed representative of real power systems. In particular, we use a set of sub-networks of the 56-bus distribution network \cite{saverioTest} (based on the IEEE 123 bus distribution feeder network \cite{IEEEdist}). We create additional subnetworks(14/30/56 bus) by aggregating nodes in the original network.
We study discrete load-curtailment problems, which are mixed integer nonconvex nonlinear optimization problems (MINLPs). We analyze a highly overloaded distribution network: a scenario that might arise just before a blackout, or after the loss of a major generator or transmission line. The goal is to curtail (reduce the consumption of) a small number of loads so that the power grid is restored to its normal operating state (bring voltages back to acceptable range). A cost is incurred for curtailing a load (typically proportional to the reduction of load). The total cost is the sum of the load-shedding costs plus a generation cost at the substation (bus $0$). The formal statement of the problem is as follows:
\begin{align*}
\mini_{\sigma,\vs,\fp} \quad &c_0\br{\injp_0,\injq_0}+\sum_{i=1}^n c_i\br{\sigma_i} 	\\
\text{ Subject to } & \eqref{eq:ACpbreac},\eqref{eq:ACpb},\eqref{eq:ACPFohm} \\
& \injp_i=\injp_i^{nom}(1-\sigma_i)+\injp_i^{red}\sigma_i,\injq_i=\injq_i^{nom}(1-\sigma_i)+\injq_i^{red}\sigma_i \\
& \sigma_i\in \{0,1\}, \vs^{L}_i \leq \vs_i \leq \vs^{U}_i
\end{align*}
The values $\injp_i^{nom},\injq_i^{nom}$ denote the nominal values of the real and reactive demands at the bus $i$. $\injp_i^{red},\injq_i^{red}$ denote the reduced (curtailed) values of the loads. $\sigma_i \in \{0,1\}$ denote the curtailment decision ($\sigma=1$ denotes curtailment). Curtailment of loads incurs a cost $c_i\br{\sigma_i}$.
   
\noindent
We run the DP algorithm (Algorithm \ref{Alg:DPMain} with update step
from Algorithms \ref{Alg:DPUpdate},\ref{Alg:DPUpdateLeaf}) on our
three test cases (14,30, and 56 buses). Ratio of DP optimum to true
optimum/upper bound, maximum constraint violation and CPU time are
studied as functions of $\epsilon$. To ensure that the results are not
artifacts of the particular test cases used, the results averaged over
$50$ instances of the problem generated by perturbing the loads in the
original problem randomly by up to $10\%$ at each bus. We show both
the mean and standard deviations of each quantity. We summarized our
observations below.

\begin{itemize}
\item Since our approach is based on a relaxation, it may produce
  infeasible solutions. However, as the radius $\epsilon$ of the
  interval discretization reduces, the degree of infeasibility, as
  measured by the maximum constraint violation, decreases. We quantify
  this dependence by taking the optimal configuration produced by the
  DP algorithm and solving the power flow equations (using Newton's
  method). We then examine if the power flow solution satisfies the
  bound constraints on voltages. Otherwise, we compute the maximum
  violation, as shown in Figures \ref{fig:14busa},\ref{fig:40busa},
  and \ref{fig:56busa}.  The results show that near-feasible
  super-optimal solutions are found by our DP algorithm consistently.

\item The other parameter of interest is the degree of
  super-optimality, which measures how close the optimal cost of the
  solution found by the DP is to the true optimal cost of the original
  problem \eqref{eq:OPFform}. For the $14$ bus network, it is feasible
  to find the true optimal cost using a brute-force search. However,
  for larger networks, we rely on the BONMIN solver to get a feasible
  solution and bound the optimality gap. The results shown in Figures
  \ref{fig:14busb},\ref{fig:40busb},\ref{fig:56busb} prove that when
  $\epsilon$ is sufficiently small, the DP algorithm optimum is within
  $.99$ of the true optimum. Note that the non-monotonic behavior of
  the optimum is due to the fact that we use an adaptive
  discretization. Even though our interval discretization gets tighter
  as $\epsilon$ gets smaller, it is not guaranteed to be a strict
  refinement of the intervals corresponding to a larger $\epsilon$.

  \item Finally, we study the dependence of the running time of the
  algorithm on $\epsilon$, in Figures
  \ref{fig:14busc},\ref{fig:40busc},\ref{fig:56busc}. The running time
  of the algorithm grows as $\epsilon$ decreases, but the plots show
  that a good optimality ratio and an acceptable error can be achieved
  with a fairly small running time of several seconds.
\end{itemize}

\begin{figure}
\centering
    \begin{subfigure}[b]{0.3\textwidth}\includegraphics[width=.95\columnwidth]{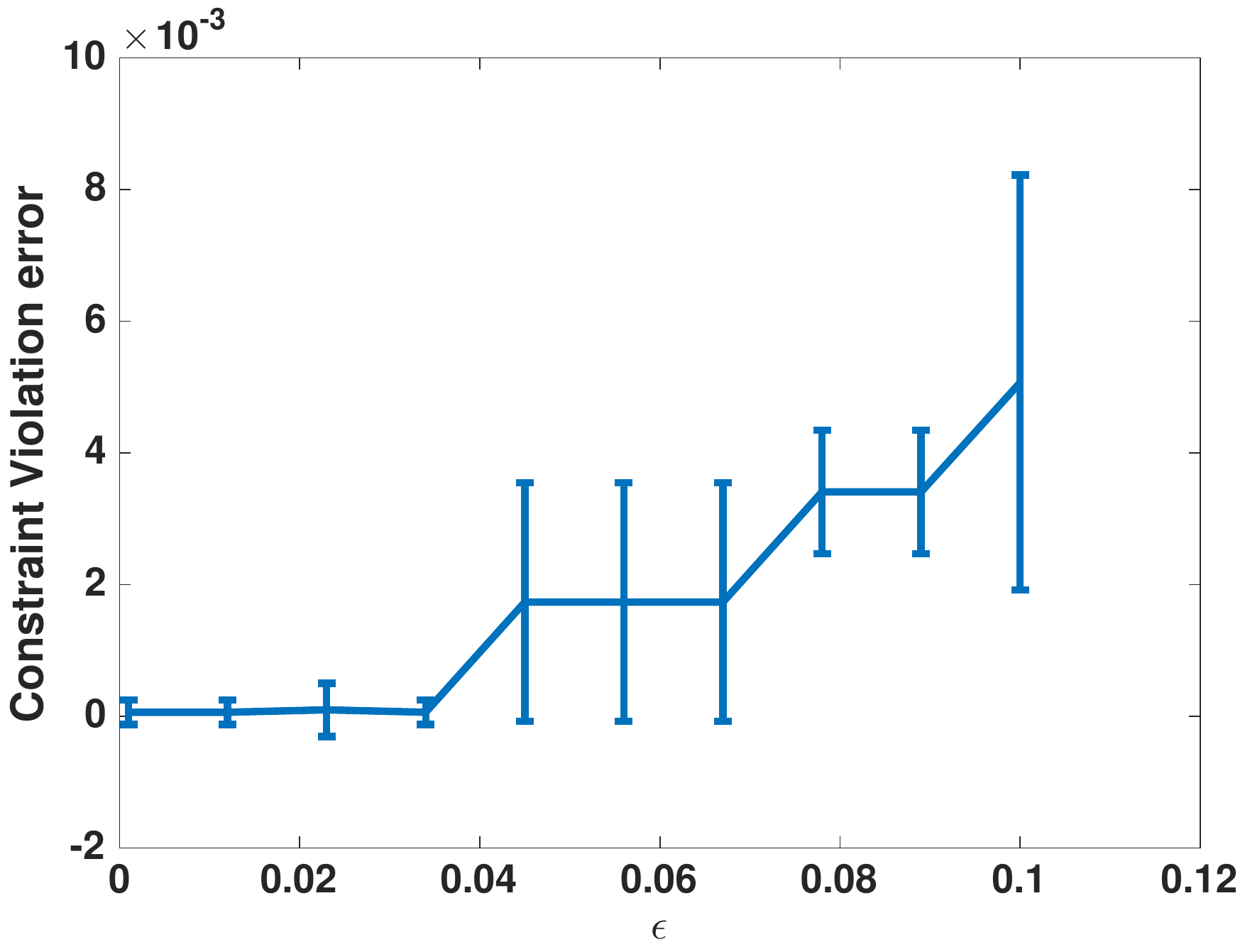}\caption{Constraint Violation}\label{fig:14busa}
    \end{subfigure}
\begin{subfigure}[b]{0.3\textwidth}
        \includegraphics[width=.95\columnwidth]{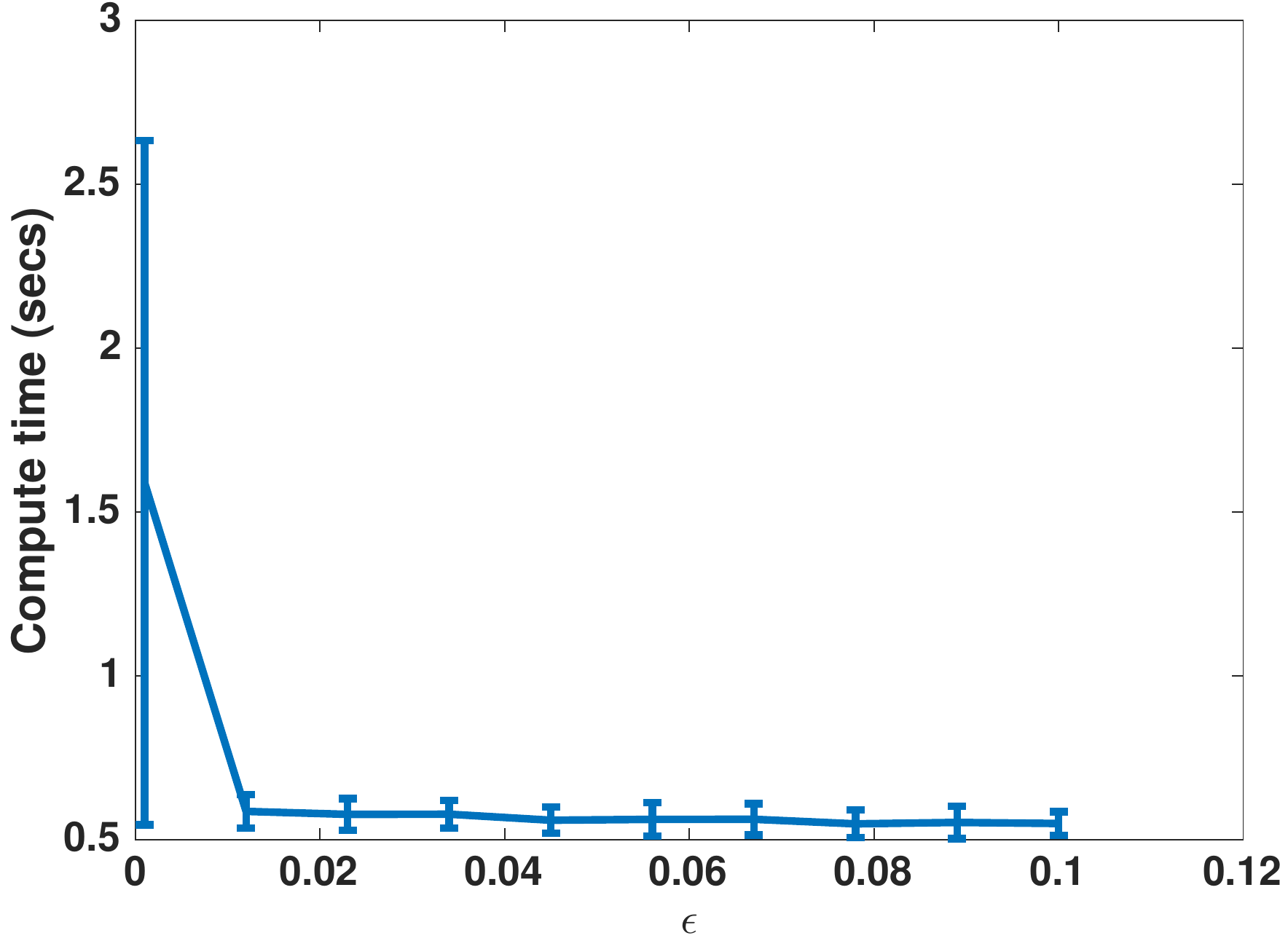}
        \caption{Computation time}\label{fig:14busb} 
        \end{subfigure}
    \begin{subfigure}[b]{0.3\textwidth}
        \includegraphics[width=.95\columnwidth]{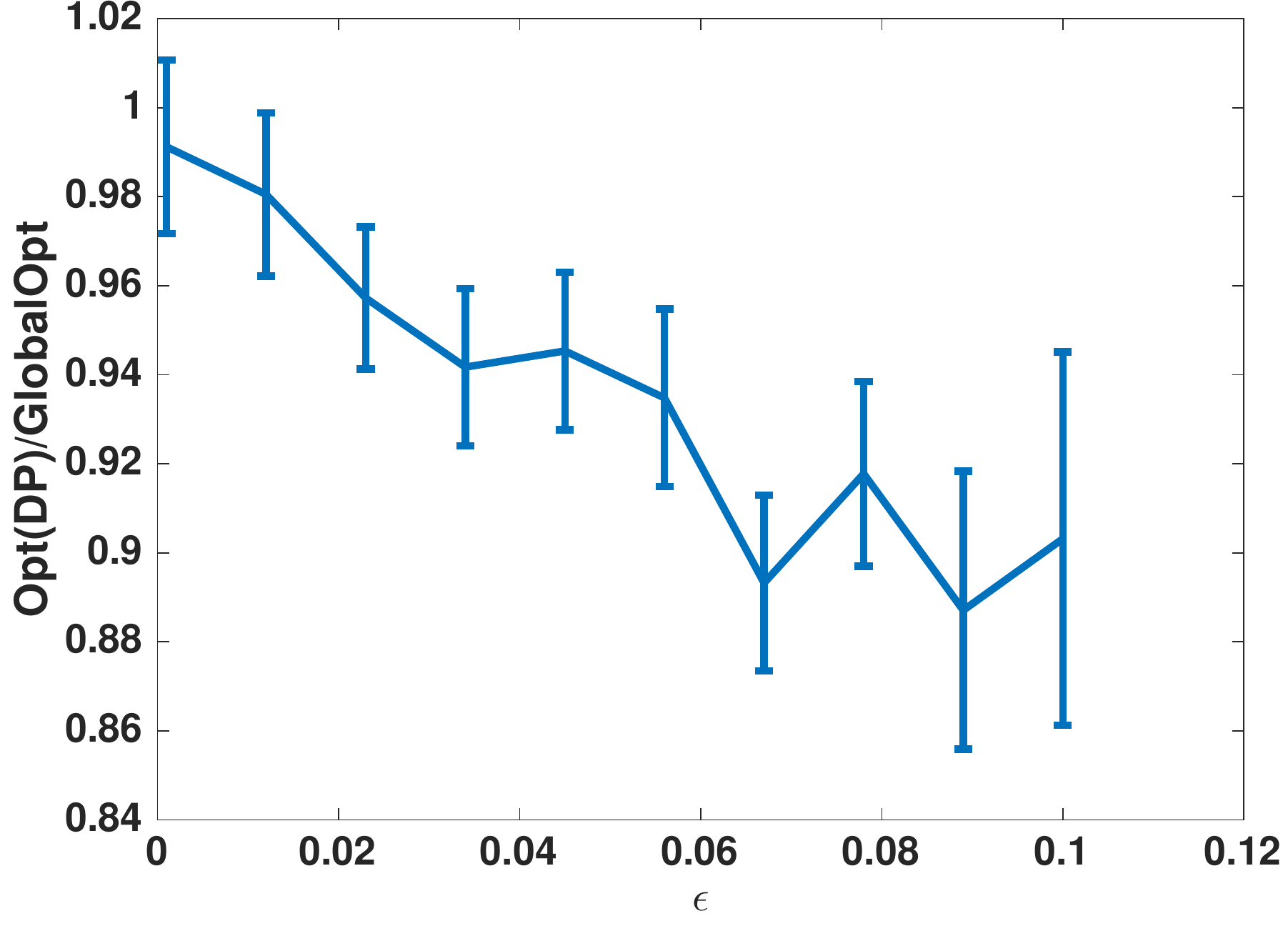}
        \caption{Superoptimality ratio}\label{fig:14busc} 
        \end{subfigure}
        \caption{14 bus network}\label{fig:ExptA}
\end{figure}
\begin{figure}
\centering
    \begin{subfigure}[b]{0.33\textwidth}\includegraphics[width=.95\columnwidth]{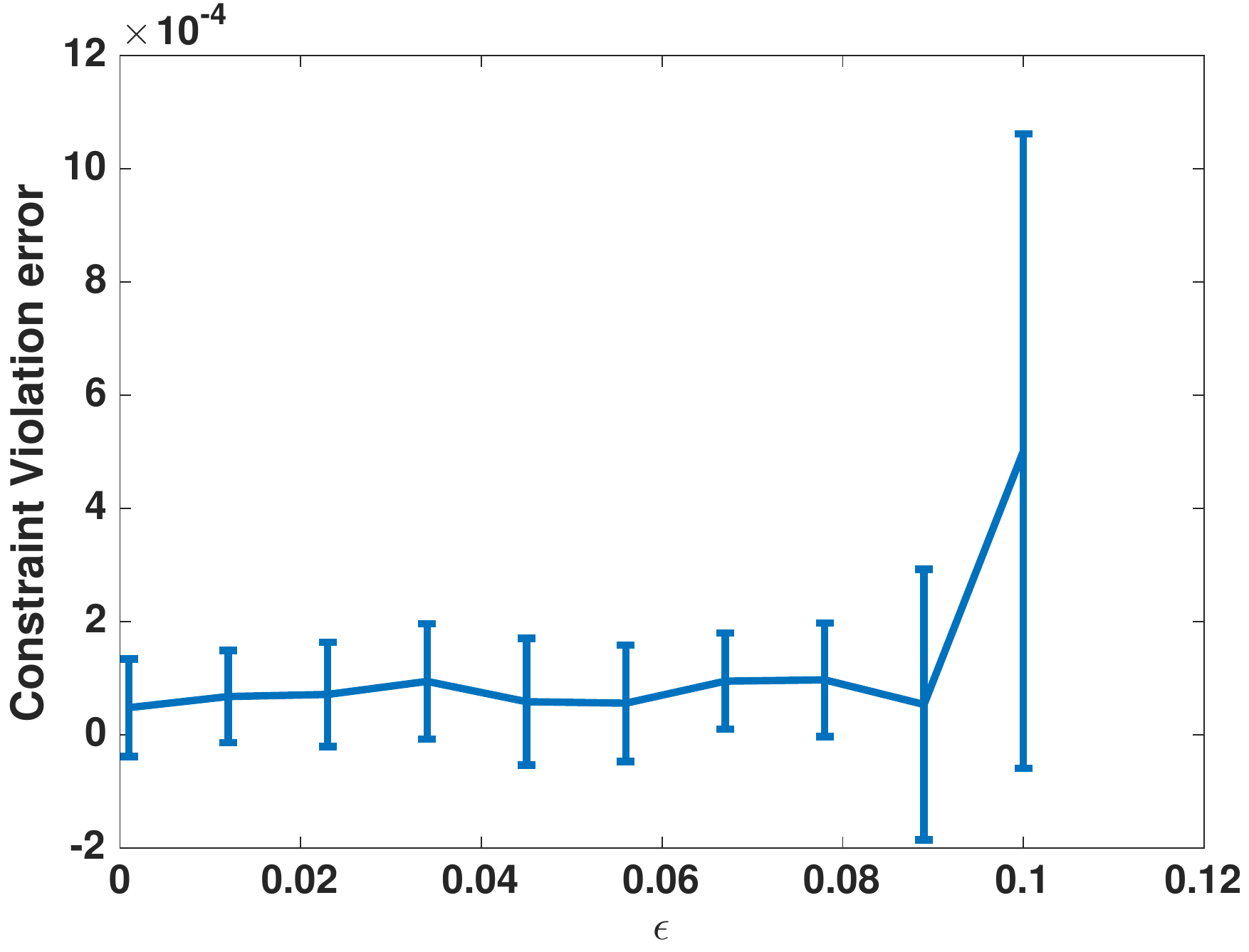}\caption{Constraint Violation}\label{fig:40busa}
    \end{subfigure}
\begin{subfigure}[b]{0.3\textwidth}
        \includegraphics[width=.95\columnwidth]{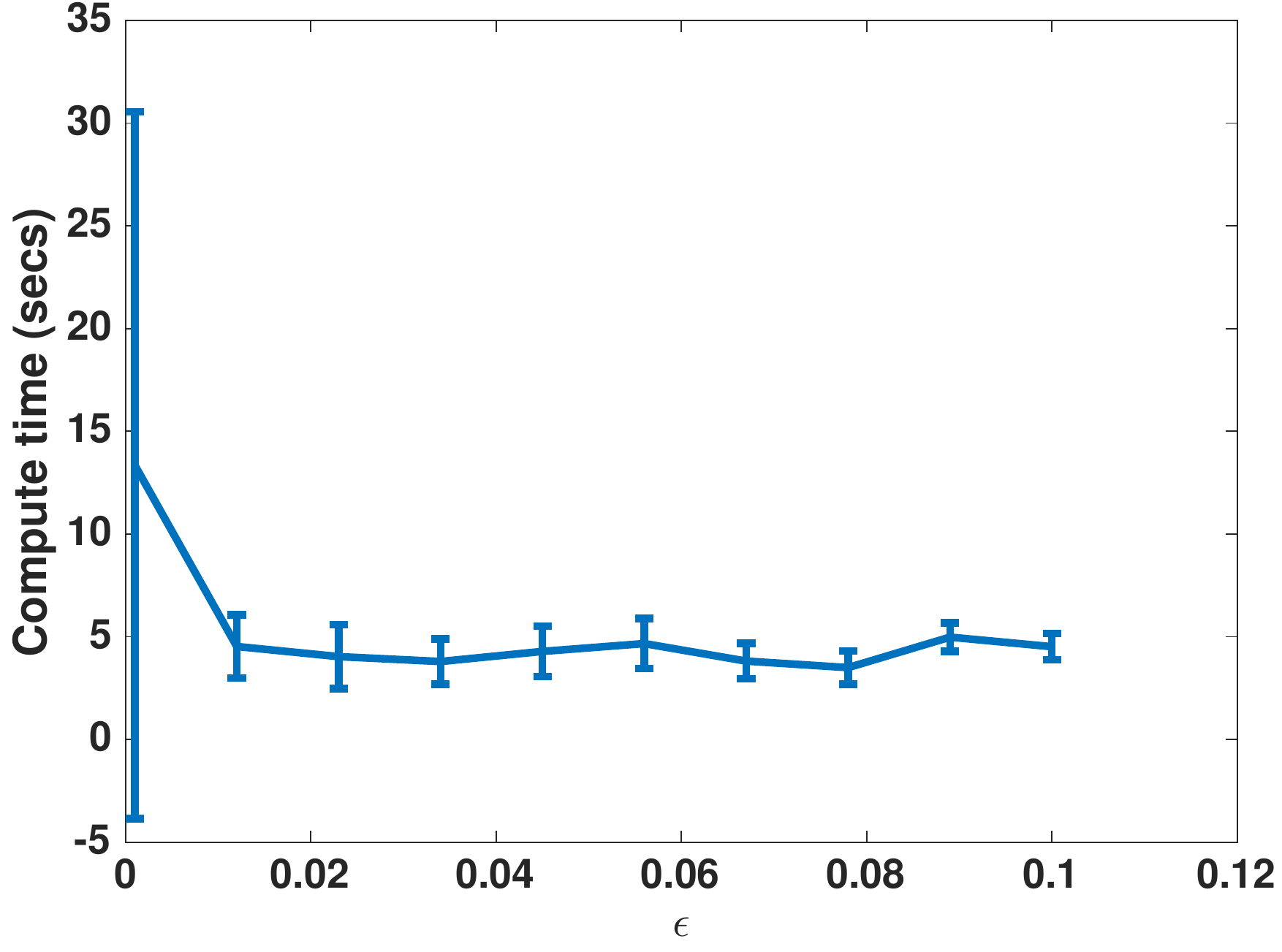}
        \caption{Computation time}\label{fig:40busb} 
        \end{subfigure}
    \begin{subfigure}[b]{0.3\textwidth}
        \includegraphics[width=.95\columnwidth]{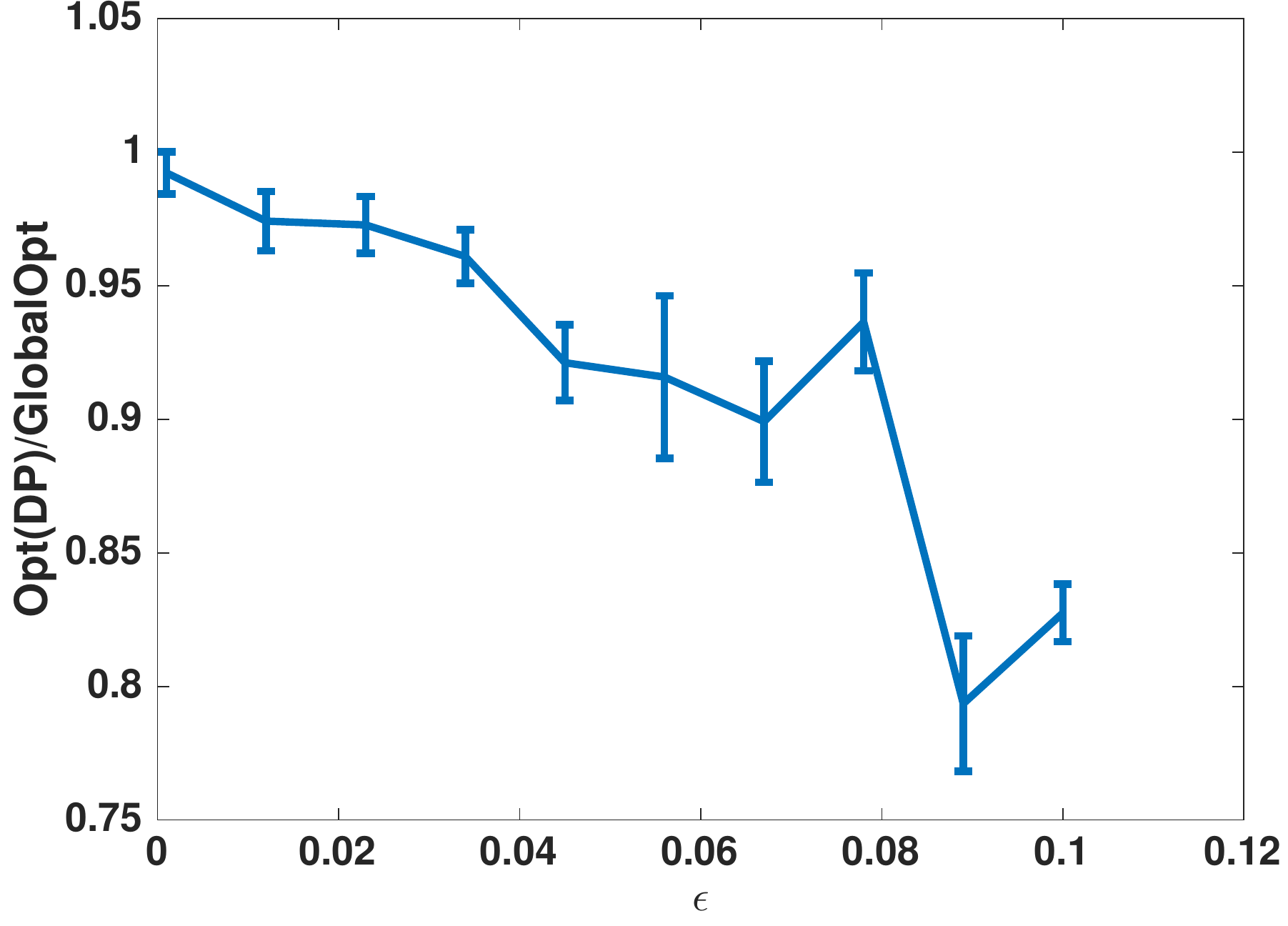}
        \caption{Superoptimality ratio}\label{fig:40busc} 
        \end{subfigure}
        \caption{40 bus network}\label{fig:ExptB}
\end{figure}

\begin{figure}
\centering
    \begin{subfigure}[b]{0.33\textwidth}\includegraphics[width=.95\columnwidth]{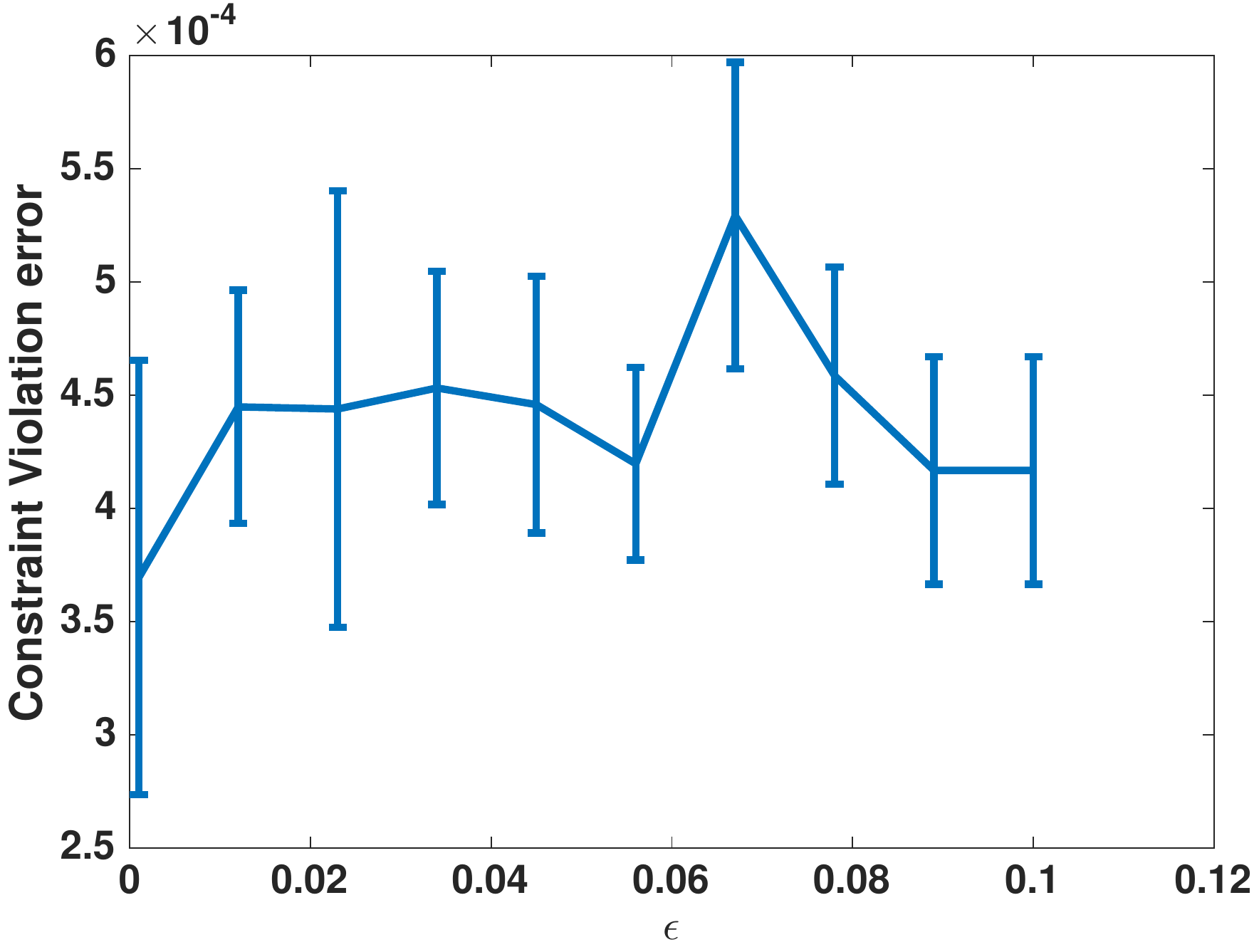}\caption{Constraint Violation}\label{fig:56busa}
    \end{subfigure}
\begin{subfigure}[b]{0.3\textwidth}
        \includegraphics[width=.95\columnwidth]{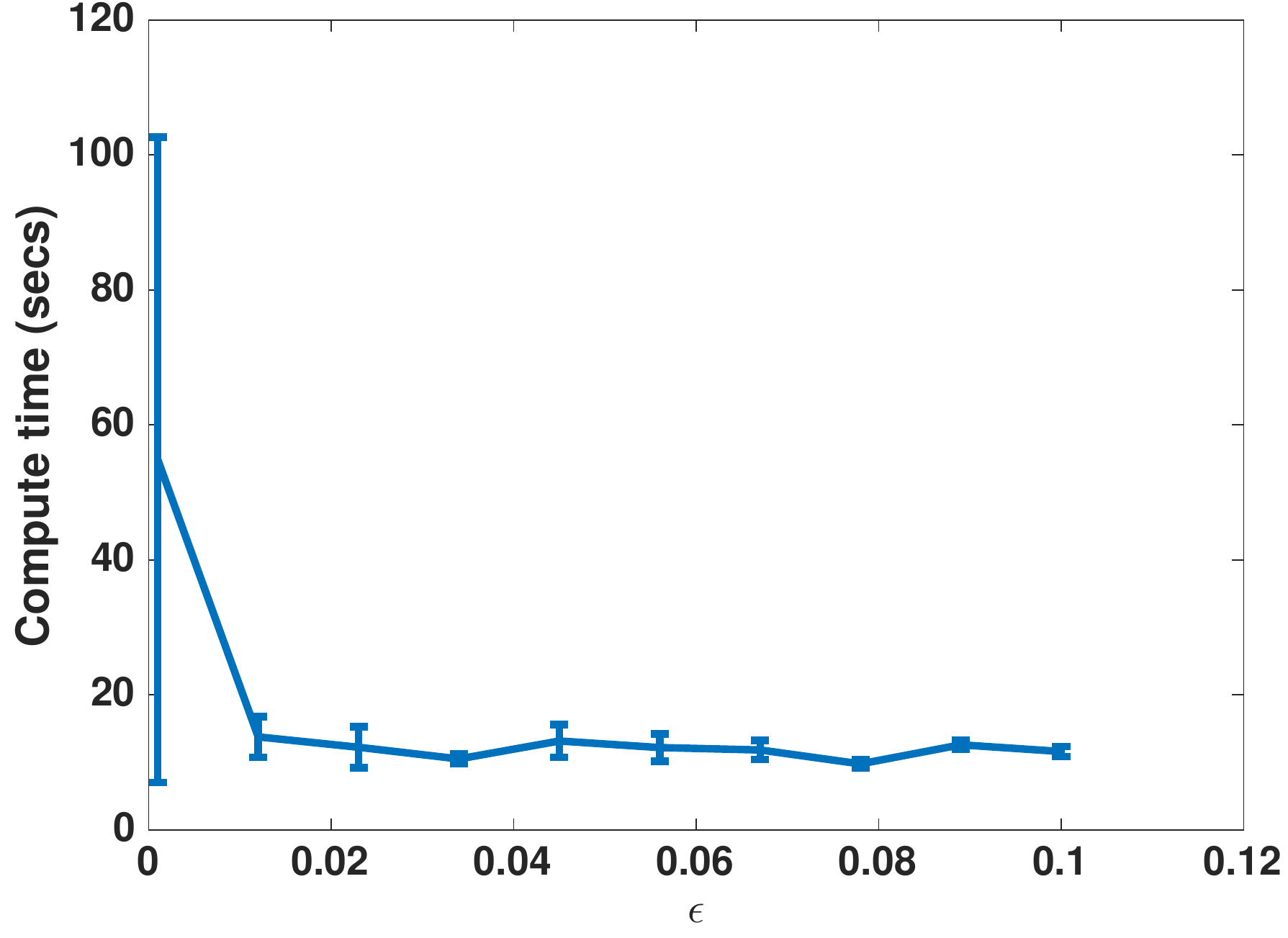}
                \caption{Computation time}\label{fig:56busb} 
        \end{subfigure}
    \begin{subfigure}[b]{0.3\textwidth}
        \includegraphics[width=.95\columnwidth]{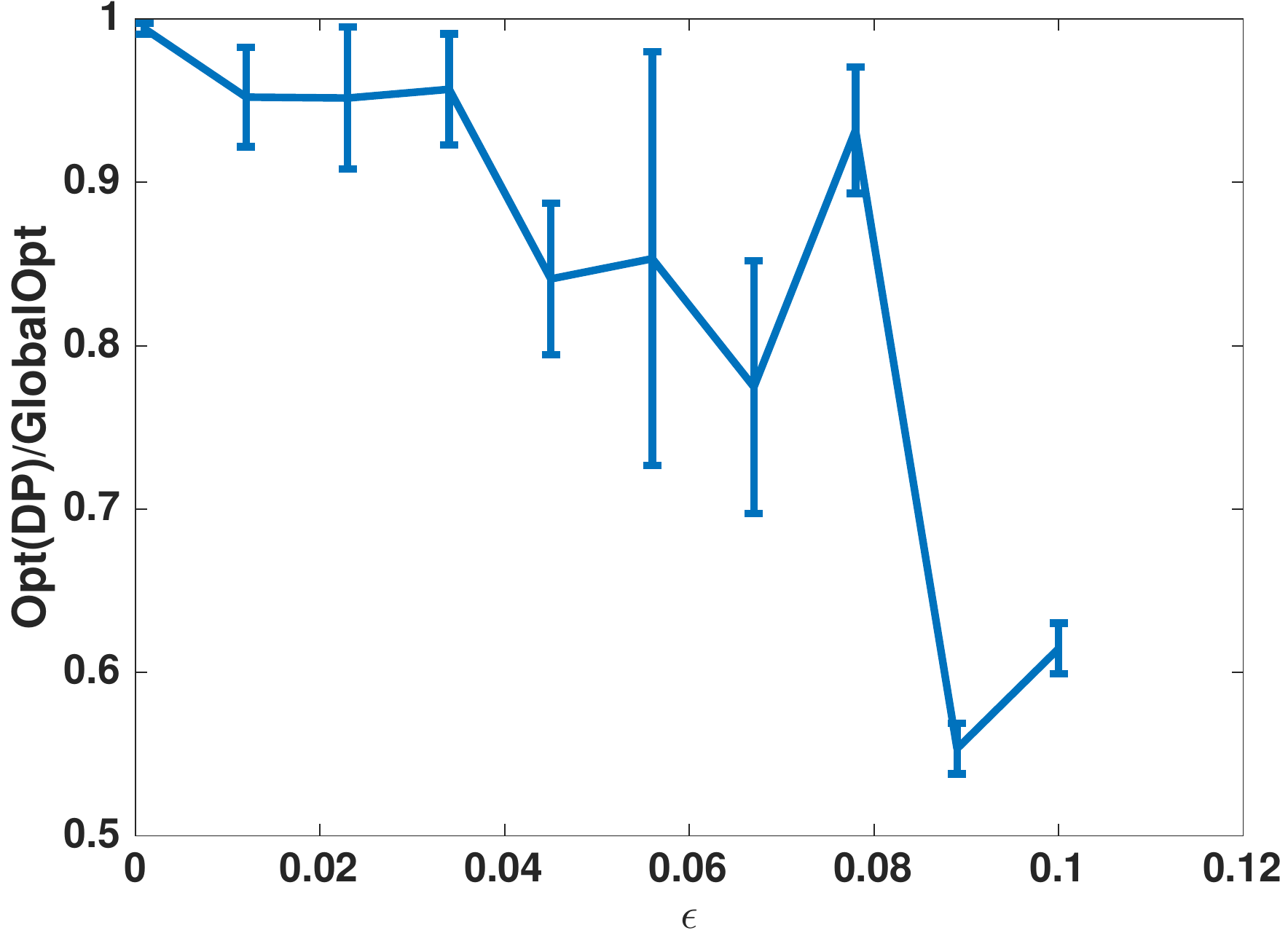}
        \caption{Superoptimality ratio}\label{fig:56busc}
        \end{subfigure}
        \caption{56 bus network}\label{fig:ExptB}
\end{figure}
	
\noindent
We also compared our algorithm to other available MINLP solvers: (1)
BONMIN \cite{bonami2008algorithmic} is a solver guaranteed to find
global optima of convex MINLPs. It can be used as a heuristic solver
for nonconvex MINLPs but with no guarantees on global optimality.  (2)
COUENNE ({\tt http://www.coin-or.org/Couenne/}) is a solver guaranteed
to find global optima of nonconvex MINLPs based on a spatial branch
and bound algorithm.  We access both solvers through the Julia
interface available via the JuMP package
\cite{DunningHuchetteLubin2015}. For the problems we studied, COUENNE
failed to converge within an acceptable time limit ($1$ hr) so we do
not report results from COUENNE. The BONMIN results are summarized in
Table \ref{tab:BONMIN}. While the BONMIN solver was faster in our
experiments, it is a heuristic solver, i.e, it is not guaranteed to
find a globally optimal solution. Indeed, BONMIN indeed fail to find
optimal solutions in the $56$ bus network. In contrast, our DP
approach always succeeds in finding a global optimum (see table
\ref{tab:DP}), although for the $56$ bus network, it requires a very
small $\epsilon$ which drives up the running time of the algorithm to
$1170s$. The reason for this behavior is that there is a super-optimal
solution that violates the voltage constraints by only $10^{-4}
\%$. This means that the discretization of voltages has to be smaller
than this for the solver to be able to recognize infeasibility of this
solution, and find the true global optimum.

\begin{table}
\begin{center}
\caption{Performance of BONMIN solver on discrete load-control: Optimality gap computed using lower bound from DP algorithm.} \label{tab:BONMIN}
\begin{tabular}{| l | l | l | l |}
\hline 
Test Case & Optimality Gap & Computation Time \\ \hline
14 bus & $0 \%$ & .1s \\ \hline 
30 bus & $0 \%$ & .2s \\ \hline 
56 bus & $0.1 \%$ & .5s \\ \hline 
\end{tabular}	
\end{center}
\end{table}

\begin{table}
\begin{center}
\caption{Performance of DP solver on discrete load-control.} \label{tab:DP}
\begin{tabular}{| l | l | l | l |}
\hline 
Solver & Test Case & Optimality Gap & Computation Time \\ \hline
DP & 14 bus & $0 \%$ & 1s \\ \hline 
DP & 30 bus & $0 \%$ & 40s \\ \hline 
DP & 56 bus & $0 \%$ & 1170s  	\\ \hline
\end{tabular}	
\end{center}	
\end{table}

\section{Conclusions}\label{sec:Conc}

We have presented a novel dynamic programming based algorithm for solving optimal power flow over tree networks. Preliminary experiments have indicated that the approach is promising and that it can solve difficult mixed integer NLP problems arising in power systems. We note that these conclusions are still preliminary and further work needs to be done to carefully test and validate the performance of this approach across a range of test problems. Overall, we envision that graphical models will be a powerful paradigm for analysis and control of power systems and other infrastructure networks. We plan to explore the following concrete directions in future work: \\
	(1) Extension to probabilistic inference problems: As solar penetration increases, the notion of security analysis (making sure that all voltages, flows, currents are within bounds) will need to be phrased in a probabilistic manner. For example, given a joint spatial distribution of solar generation at various points in the network, compute the probability that a given physical quantity (voltage/current/flow) deviates beyond its acceptable bounds. This problem can be phrased as the sum-product analog of the problem solved here. \\
	(2) Extensions to loopy graphs: There are several possibilities for extending the algorithms presented here to loopy graphs. The most straightforward extensions would be based on junction trees \cite{koller2009probabilistic} (cluster nodes into supernodes to form a tree) or on cutset conditioning \cite{dechter2003constraint} (fix values of variables on a cutset of the graph, and given for each fixed value, use inference on the remaining tree-structured graph). Another route is to use loopy belief propagation or the corresponding Linear Programming relaxation of the inference problem \cite{koller2009probabilistic}, and subsequent hierarchies of relaxations, in the spirit of \cite{sontag2010approximate}\cite{johnson2008convex}. \\
	(3) Parameterized messages: We represented messages with piecewise-constant approximations. Another option is to use a parameterized representation (polynomial/piecewise linear/piecewise polynomial for ex). An interesting related development is \cite{gamarnik2012belief}, where the authors show that belief propagation with piecewise-linear messages is guaranteed to find the global optimum of a certain special minimum cost flow problem in polynomial time. Extending this to ACOPF is another promising direction for future work.
\bibliographystyle{spmpsci}      

\bibliography{Ref}

\begin{thebibliography}{10}
\providecommand{\url}[1]{{#1}}
\providecommand{\urlprefix}{URL }
\expandafter\ifx\csname urlstyle\endcsname\relax
  \providecommand{\doi}[1]{DOI~\discretionary{}{}{}#1}\else
  \providecommand{\doi}{DOI~\discretionary{}{}{}\begingroup
  \urlstyle{rm}\Url}\fi

\bibitem{IEEEdist}
https://ewh.ieee.org/soc/pes/dsacom/testfeeders/

\bibitem{saverioTest}
https://github.com/saveriob/approx-pf

\bibitem{baran1989optimal}
Baran, M.E., Wu, F.F.: Optimal capacitor placement on radial distribution
  systems.
\newblock Power Delivery, IEEE Transactions on \textbf{4}(1), 725--734 (1989)

\bibitem{bienstock2015lp}
Bienstock, D., Munoz, G.: Lp approximations to mixed-integer polynomial
  optimization problems.
\newblock ArXiv e-prints  (2015)

\bibitem{bienstock2015strong}
Bienstock, D., Verma, A.: Strong np-hardness of ac power flows feasibility.
\newblock arXiv preprint arXiv:1512.07315  (2015)

\bibitem{bonami2008algorithmic}
Bonami, P., Biegler, L.T., Conn, A.R., Cornu{\'e}jols, G., Grossmann, I.E.,
  Laird, C.D., Lee, J., Lodi, A., Margot, F., Sawaya, N., et~al.: An
  algorithmic framework for convex mixed integer nonlinear programs.
\newblock Discrete Optimization \textbf{5}(2), 186--204 (2008)

\bibitem{BaseQCQP}
Bose, S., Gayme, D., Chandy, K., Low, S.: Quadratically constrained quadratic
  programs on acyclic graphs with application to power flow.
\newblock Control of Network Systems, IEEE Transactions on \textbf{2}(3),
  278--287 (2015).
\newblock \doi{10.1109/TCNS.2015.2401172}

\bibitem{carpentier1962contribution}
Carpentier, J.: Contribution to the economic dispatch problem.
\newblock Bulletin de la Societe Francoise des Electriciens \textbf{3}(8),
  431--447 (1962)

\bibitem{coffrin2015strengthening}
Coffrin, C., Hijazi, H.L., Van~Hentenryck, P.: Strengthening convex relaxations
  with bound tightening for power network optimization.
\newblock In: Principles and Practice of Constraint Programming, pp. 39--57.
  Springer (2015)

\bibitem{dechter2003constraint}
Dechter, R.: Constraint processing.
\newblock Morgan Kaufmann (2003)

\bibitem{DunningHuchetteLubin2015}
Dunning, I., Huchette, J., Lubin, M.: {JuMP}: {A} modeling language for
  mathematical optimization.
\newblock arXiv:1508.01982 [math.OC]  (2015).
\newblock \urlprefix\url{http://arxiv.org/abs/1508.01982}

\bibitem{gamarnik2012belief}
Gamarnik, D., Shah, D., Wei, Y.: Belief propagation for min-cost network flow:
  Convergence and correctness.
\newblock Operations research \textbf{60}(2), 410--428 (2012)

\bibitem{gan2015exact}
Gan, L., Li, N., Topcu, U., Low, S.H.: Exact convex relaxation of optimal power
  flow in radial networks.
\newblock Automatic Control, IEEE Transactions on \textbf{60}(1), 72--87 (2015)

\bibitem{johnson2008convex}
Johnson, J.K.: Convex relaxation methods for graphical models: Lagrangian and
  maximum entropy approaches.
\newblock Ph.D. thesis, Massachusetts Institute of Technology (2008)

\bibitem{kearns2001graphical}
Kearns, M., Littman, M.L., Singh, S.: Graphical models for game theory.
\newblock In: Proceedings of the Seventeenth conference on Uncertainty in
  artificial intelligence, pp. 253--260. Morgan Kaufmann Publishers Inc. (2001)

\bibitem{koller2009probabilistic}
Koller, D., Friedman, N.: Probabilistic graphical models: principles and
  techniques.
\newblock MIT press (2009)

\bibitem{kschischang2001factor}
Kschischang, F.R., Frey, B.J., Loeliger, H.A.: Factor graphs and the
  sum-product algorithm.
\newblock Information Theory, IEEE Transactions on \textbf{47}(2), 498--519
  (2001)

\bibitem{lavaei2013geometry}
Lavaei, J., Tse, D., Zhang, B.: Geometry of power flows and optimization in
  distribution networks  (2013)

\bibitem{PascalHardness}
Lehmann, K., Grastien, A., Van~Hentenryck, P.: Ac-feasibility on tree networks
  is np-hard.
\newblock Power Systems, IEEE Transactions on \textbf{31}(1), 798--801 (2016).
\newblock \doi{10.1109/TPWRS.2015.2407363}

\bibitem{low2014convex}
Low, S.H.: Convex relaxation of optimal power flow, part i: Formulations and
  equivalence.
\newblock arXiv preprint arXiv:1405.0766  (2014)

\bibitem{low2014convexb}
{Low}, S.H.: {Convex Relaxation of Optimal Power Flow, Part II: Exactness}.
\newblock ArXiv e-prints  (2014)

\bibitem{pai2014computer}
Pai, M., Chatterjee, D.: Computer techniques in power system analysis.
\newblock McGraw-Hill Education (India) (2014)

\bibitem{sojoudi2014exactness}
Sojoudi, S., Lavaei, J.: Exactness of semidefinite relaxations for nonlinear
  optimization problems with underlying graph structure.
\newblock SIAM Journal on Optimization \textbf{24}(4), 1746--1778 (2014)

\bibitem{sontag2010approximate}
Sontag, D.A.: Approximate inference in graphical models using lp relaxations.
\newblock Ph.D. thesis, Massachusetts Institute of Technology (2010)

\bibitem{BergenVittal}
Vittal, V., Bergen, A.R.: Power systems analysis.
\newblock Prentice Hall (1999)

\bibitem{wainwright2008graphical}
Wainwright, M.J., Jordan, M.I.: Graphical models, exponential families, and
  variational inference.
\newblock Foundations and Trends{\textregistered} in Machine Learning
  \textbf{1}(1-2), 1--305 (2008)

\end{thebibliography}

\appendix

\section{Appendix}\label{sec:App}
\subsection{Convex Envelopes of quadratic and bilinear terms}
The nonlinearities appearing in power flow are of the form $x^2$ or $xy$ for some variables $x,y$. We use the following convex envelopes as relaxations of these nonlinear terms: 
\begin{subequations}
\begin{align}
\mathrm{SqRel}\br{y,[\lb{y},\ub{y}]}=\left\{x: \begin{array}{ll}
                  x &\geq y^2\\
                  x &\leq \br{\lb{y}+\ub{y}}y-\lb{y}\ub{y}
                \end{array}	\right\} \\
\mathrm{McCormick}\br{y,z,[\lb{y},\ub{y}],[\lb{z},\ub{z}]}=\left\{x: \begin{array}{ll}
                  x &\geq \lb{y}z + \lb{z}y - \lb{y}\lb{z}\\
                  x &\geq \ub{y}z + \ub{z}y - \ub{y}\ub{z}\\
				  x &\leq \lb{y}z - \ub{z}y + \lb{y}\ub{z}\\
				  x &\leq \lb{z}y - \ub{y}z + \lb{z}\ub{y}	
                \end{array}	\right\}              
\end{align}	
\end{subequations}

\subsection{Interval DP Relaxation}\label{sec:AppRelax}
In this section, we describe a concrete implementation of the interval relaxation procedure \eqref{eq:PropBoundRelax}:
\begin{subequations}
\begin{align}
 &\Ext_{\inj_i,\{\fp_k,\vs_k,\mathrm{Sq}^{\fpp}_k,,\mathrm{Sq}^{\fpp}_k,\mathrm{Prod}^{\vs,\cur}_k\}_{k\in \Chi{i}\cup\{i\}}}  \{\vs_i,\fpp_i,\fpq_i,a_i\br{t}\injp_i+b_i\br{t}\injq_i+c_i\br{t}\} \\
 & \text{Subject to } \nonumber\\
 & \fpp_i = \injp_i+\sum_{k \in \Chi{i}}\br{\fpp_k-\cur_k\zr_k} \label{eq:RelaxDPaa}\\
  & \fpq_i = \injq_i+\sum_{k \in \Chi{i}}\br{\fpq_k-\cur_k\zi_k} \label{eq:RelaxDPab}\\
& \vs_i = \vs_k+\cur_k\br{\zr_k^2+\zi_k^2}-2\br{\fpp_k\zr_k+\fpq_k\zi_k},k \in \Chi{i} \label{eq:RelaxDPa}\\
& \sqrt{\fpp_k^2+\fpq_k^2} \leq \sqrt{\vs_k\cur_k}\quad k\in \Chi{i}\cup\{i\} \label{eq:RelaxDPb}\\
& \mathrm{Sq}^{\fpp}_k+\mathrm{Sq}^{\fpq}_k=\mathrm{Prod}^{\vs,\cur}_k \quad k\in \Chi{i}\cup\{i\}\label{eq:RelaxDPc}\\
& \mathrm{Sq}^{\fpp}_k \in \mathrm{SqRel}\br{\fpp_k,\fpp\br{\Int_k}} \quad k\in \Chi{i}\cup\{i\} \label{eq:RelaxDPd}\\
& \mathrm{Sq}^{\fpq}_k \in \mathrm{SqRel}\br{\fpq_k,\fpq\br{\Int_k}} \quad k\in \Chi{i}\cup\{i\} \label{eq:RelaxDPe}\\
& \mathrm{Prod}^{\vs,\cur}_k \in \mathrm{McCormick}\br{\vs_k,\cur_k,\vs\br{\Int_k},\cur\br{\Int_k}}\quad k\in \Chi{i}\cup\{i\}	\label{eq:RelaxDPf} \\
&\vs_k \in \vs\br{\Int_k}, \cur_k \in \cur\br{\Int_k},\fpp_k\in \fpp\br{\Int_k},\fpq_k\in \fpq\br{\Int_k},k\in\{i\}\cup\Chi{i} \\
& \injp_i \in [\lb{\injp_i}\br{t},\ub{\injp_i}\br{t}],\injq_i \in [\lb{\injq_i}\br{t},\ub{\injq_i}\br{t}]
\end{align}\label{eq:RelaxDP}	
\end{subequations}

This requires solution of a small number of SOCPs within each DP update (specifically $10$ SOCPs in $6d$ variables where $d$ is the maximum degree of a node in the tree - note that we can always choose $d\leq 2$ by modifying the original problem as in lemma \ref{lem:DegLem}). Note also that as the intervals $\Int_k$ get smaller, the relaxation gets tighter - this is formalized in the lemma below:

\begin{lemma}\label{lem:AppRelax}
The relaxation defined by \eqref{eq:RelaxDP} is a valid interval relaxation 	and satisfies the conditions of the definition \ref{def:IntervalRelax}.
\end{lemma}
\begin{proof}
Through this proof, we use $\eta\br{M}$ to refer to some constant that depends on the number $M$ from assumption \ref{assump:B}. Properties \eqref{eq:IntervalRelaxDefb},\eqref{eq:IntervalDefb} are obvious since \eqref{eq:RelaxDP} is a valid relaxation of the problem \eqref{eq:PropBound}. The property \eqref{eq:IntervalDefa} follows from the tightness of the McCormick relaxation. If $\Int_i,\br{\Int_k	}_{k\in \Chi{i}}$ are of radius at most $\epsilon$, we know that 
\[\lb{\vs_k}\cur_k+\lb{\cur_k}\vs_k-\lb{\vs_k}\lb{\cur_k} \leq \mathrm{Prod}^{\vs,\cur}_k \leq \ub{\vs_k}\cur_k+\lb{\cur_k}\vs_k-\ub{\vs_k}\lb{\cur_k}\]
so that the range of $\mathrm{Prod}^{\vs,\cur}_k$ is of size at most 
\[\br{\ub{\vs_k}-\lb{\vs_k}}\br{\cur_k-\lb{\cur_k}}\leq \br{\ub{\vs_k}-\lb{\vs_k}}\br{\ub{\cur_k}-\lb{\cur_k}}\leq \eta \epsilon\] since $\cur_k$ has upper and lower bounds depending on the problem data. Thus, we know that $\mathrm{Prod}_k^{\vs,\cur}$ is at most $\eta\br{M} \epsilon$ away from $\vs_k\cur_k$. 
Similarly, $\mathrm{Sq}^{\fpp_k}$ is at most $\eta\br{M}\epsilon$ away from $\fpp_k^2$, $\mathrm{Sq}^{\fpq_k}$ is at most $\eta\br{M}\epsilon$ away from $\fpq_k^2$. Combining these results, we get the that \eqref{eq:IntervalDefa} holds.
\end{proof}

\subsection{Bound-tightening procedure}\label{sec:Relax}
We use a scheme similar to the one described in \cite{coffrin2015strengthening} that infers bounds on the variables $\vs_i,\fpp_i,\fpq_i\cur_i$ given the constraints in the ACOPF problem \eqref{eq:DPMain} . Let $\mathrm{Conv}\br{\Inj_i}$ denote the convex hull of the set $\Inj_i$. We use a convex relaxation of the constraints \eqref{eq:OPFform} that depends on the variable bounds. We then iterate this procedure where we use a relaxation to infer tighter bounds, and then tighten the relaxation using the inferred bounds. In practice, we find that the procedure converges in a few iterations typically to a stable set of bounds.

The nonconvex constraint $\fpp_k^2+\fpq_k^2=\vs_k\cur_k$ can be relaxed to a convex constraint: $\sqrt{\fpp_k^2+\fpq_k^2}\leq\sqrt{\cur_k\vs_k}$. This can be tightened by replacing the nonlinear terms in the equation $\powb{\fpp_k}{2}+\powb{\fpq_k}{2}=\vs_k\cur_k$ with their McCormick envelopes. Plugging all this into a single formulation, we obtain:
\begin{subequations}
\begin{align}
\Ext_{\vs,\cur,\fp,\inj,\mathrm{Sq}^{\fpp},\mathrm{Sq}^{\fpq},\mathrm{Prod}^{\vs,\cur}} &\quad  \{\vs_i,\cur_i,\fpp_i,\fpq_i\}_{i=1}^n \\
\text{ Subject to } & \fp_i-\sum_{k \in \Chi{i}} \br{\fp_k-\cur_kz_k} \in \mathrm{Conv}\br{\Inj_i}, i \in \{0,\ldots,n\} \\
& \vs_i = \vs_k+\cur_k|z_k|^2-2\br{\fpp_k\zr_k+\fpq_k\zi_k}, i\in \{0,\ldots,n\}, k \in \Chi{i} \\
& \sqrt{\fpp_k^2+\fpq_k^2} \leq \sqrt{\vs_k\cur_k}, k\in \{1,\ldots,n\} \\
& \mathrm{Sq}^{\fpp}_k+\mathrm{Sq}^{\fpq}_k=\mathrm{Prod}^{\vs,\cur}_k \quad k\in \{1,\ldots,n\}\\
& \mathrm{Sq}^{\fpp}_k \in \mathrm{SqRel}\br{\fpp_k,[\lb{\fpp_k},\ub{\fpp_k}]}, k\in \{1,\ldots,n\}\\
& \mathrm{Sq}^{\fpq}_k \in \mathrm{SqRel}\br{\fpq_k,[\lb{\fpq_k},\ub{\fpq_k}]}, k\in \{1,\ldots,n\}\\
& \mathrm{Prod}^{\vs,I}_k \in \mathrm{McCormick}\br{\vs_k,\cur_k,[\lb{\vs_k},\ub{\vs_k}],[\lb{\cur_k},\ub{\cur_k}]}, k\in \{1,\ldots,n\} 
\end{align}	\label{eq:BoundTighten}
\end{subequations}
Each minimum/maximum value involves solving a Second Order Cone Program (SOCP) and can be done in parallel over the variables involved. This entire procedure can be viewed as a mapping:

\[{\begin{pmatrix} \lb{\vs} & \ub{\vs} & \lb{\fp} & \ub{\fp} & \lb{\cur} & \ub{\cur} \end{pmatrix}}_t \mapsto {\begin{pmatrix} \lb{\vs} & \ub{\vs} & \lb{\fp} & \ub{\fp} & \lb{\cur} & \ub{\cur} \end{pmatrix}}_{t+1}\]

We iterate this mapping until there the improvement in bounds is smaller than some threshold. The obtained bounds are used to redefine the domains $\Xcal_i$ for each of the variables.

\subsection{Proof of Lemma \ref{lem:DegLem}}\label{sec:AppDeg}
We describe a transformation that takes a node $i$ with $m$ children and adds at most $r=\lceil\log_2\br{m}\rceil$ additional buses to create a new network where each node has at most $2$ children. 
\noindent
We add children in ``levels'' $p=1,\ldots,r$: At level $1$, we add children $c_{10},c_{11}$  connected to bus $i$ by $0$-impedance transmission lines. We have the following constraints between $i$ and its children ${\Chi{i}}^{\prime}=\{c_{10},c_{11}\}$:
\begin{align*}
\fp_i & =\inj_i+\sum_{k \in {\Chi{i}}^{\prime}}\fp_k \\
\vs_i & =\vs_k, k \in \Chi{i}^{\prime}	\\
\cur_k\vs_k &= |\fp_k|^2, k \in {\Chi{i}}^{\prime}
\end{align*}
\noindent
At any level $p \leq r$, all nodes are of the form $c_{i_1\ldots i_p}$. We add its children $\Chi{c_{i_1\ldots i_p}}=\{c_{i_1\ldots i_p 0},c_{i_1\ldots i_p 1}\}$ connected to it by $0$-impedance lines with the constraints:
\begin{align*}
\fp_{c_{i_1\ldots i_p}} & =\sum_{k \in \Chi{c_{i_1\ldots i_p}}}\fp_k \\
\vs_{c_{i_1\ldots i_p}} & =\vs_k, k \in \Chi{c_{i_1\ldots i_p}}	\\
\cur_k\vs_k &= |\fp_k|^2, k \in \Chi{c_{i_1\ldots i_p}}
\end{align*}

At the final level $p=r-1$, every node is of the form $c_{i_1\ldots i_{r-1}}$ and its children are picked from the set of original children $\Chi{i}$. One way of doing this is to assign children in order: $\Chi{c_{10\ldots 0}}=\{c_1,c_2\},\Chi{c_{10\ldots 1}}=\{c_3,c_4\},\ldots$. Then, we add the balance equations:
\begin{align*}
\fp_{c_{i_1i_2\ldots i_{r-1}}} & =\sum_{k \in \Chi{c_{i_1i_2\ldots i_{r-1}}}}\br{\fp_k} \\
\vs_i & =\vs_k, k \in \Chi{i}^{\prime}	\\
\cur_k\vs_k &= |\fp_k|^2, k \in {\Chi{i}}^{\prime}
\end{align*}

Adding the power balance equations at all the intermediate buses, we recover the original power balance condition
\[\inj_i=\sum_{k \in \Chi{i}} \br{\fp_k-z_k\cur_k} \]
Further, we have that $\vs_i=\vs_{c_{i_1i_2\ldots i_{p}}}$ for every $1\leq p \leq r-1$.

\begin{figure}
\begin{center}
\includegraphics[width=.75\textwidth]{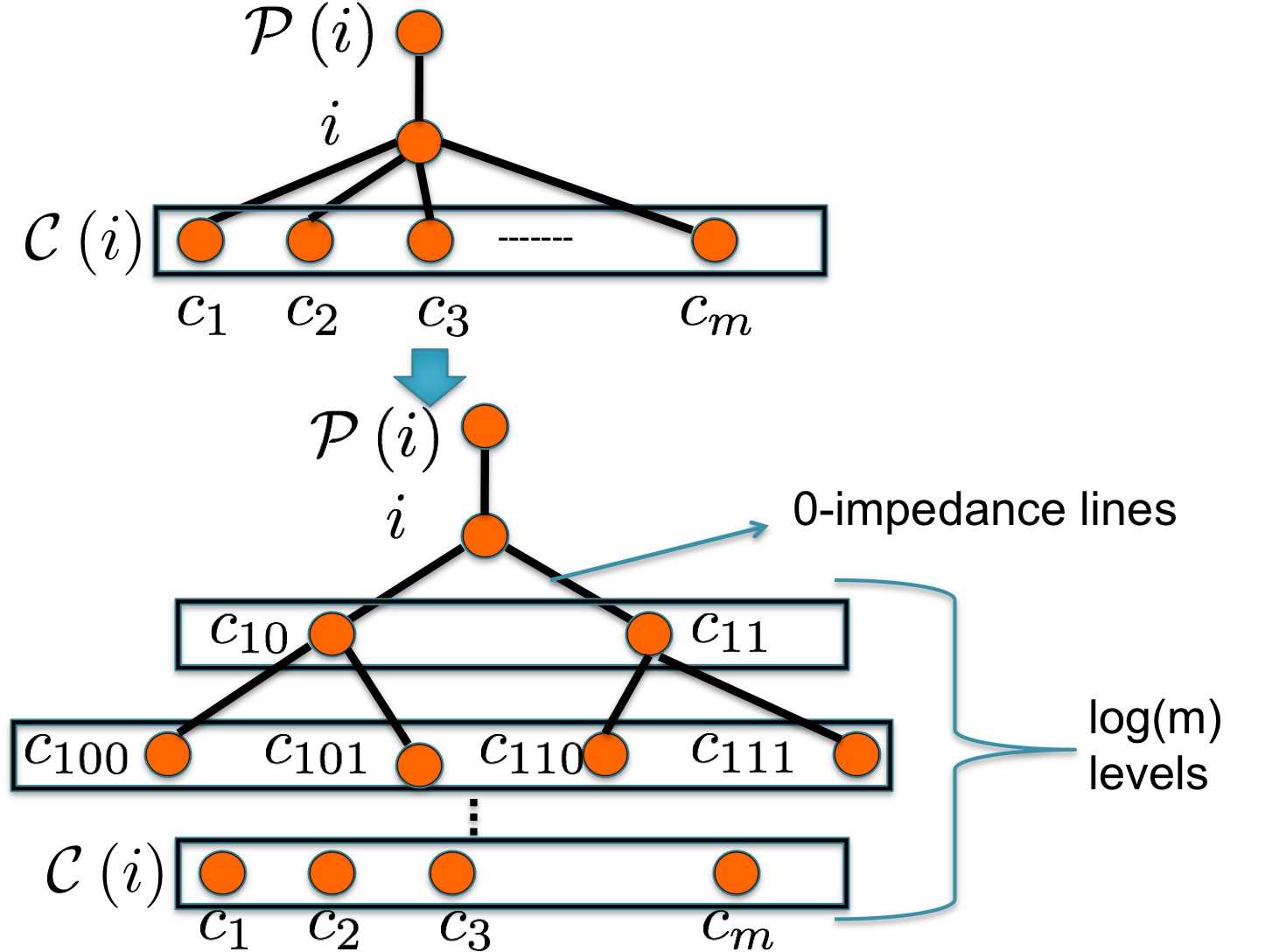}	 	
\end{center}
\end{figure}

\subsection{Proof of Theorem \ref{thm:ApproxMain}}\label{sec:AppMain}
\begin{proof}
Through this proof, we will use $\zeta\br{M}$ to denote an arbitrary function of $M$. The proof of the algorithm breaks down to three key statements:

\begin{itemize}
	\item[1] The size of the messages $|\eta^i|$ is bounded by $\frac{\zeta\br{M}}{\epsilon^3}$. 
	\item[2] For each message, $\mathrm{Rad}\br{\xs^i\br{t}}\leq \zeta\br{M}\epsilon,t=1\ldots,|\xs^i|$.
	\item[3] In the interval DP update Algorithm \ref{Alg:DPUpdate} makes at most $\frac{\zeta\br{M}}{\epsilon^5}$ calls to the $\mathrm{PropBound}$ routine. 
\end{itemize}

\emph{Proof of 1,2}\\
For the leaf nodes, the size of the messages is bounded by $|\Inj_i|\br{\frac{\ub{\vs_i}-\lb{\vs_i}}{\epsilon}+1}\leq \frac{\zeta\br{M}}{\epsilon}$ (since $\ub{\vs}_i\leq M, \lb{\vs_i} \geq \frac{1}{M},|\Inj_i|\leq M$). Since $0<\epsilon<1$, this is smaller than $\frac{\zeta\br{M}}{\epsilon^3}$. Further, since $\ub{\injp_i}\br{t}-\injp_i\br{t},\ub{\injq_i}\br{t}-\injq_i\br{t}\leq \frac{1}{M}$, we know that $\mathrm{Rad}\br{\xs^i\br{t}}\leq \epsilon$ for each $t=1,\ldots,|\xs^i|$. 
For non-leaf node, the size of the messages is bounded by the size of $|\mathrm{Partition}\br{\Xcal_i,\epsilon}|\leq\frac{\zeta\br{M}}{\epsilon^3}$. Further, since $\xs^i\br{t} \subset \Int_i \in \mathrm{Partition}\br{\Xcal_i,\epsilon}$, we know that $\mathrm{Rad}\br{\xs^i\br{t}}\leq \epsilon$.

Thus, using \eqref{eq:IntervalRelaxDef}, the error of the PF equations $|\injp_i-\injp_i\br{\xs_i,\xs_{\Chi{i}}}|,|\injq_i-\injq_i\br{\xs_i,\xs_{\Chi{i}}}|,|\vs_i-\vs_i\br{\xs_k}|$ can be bounded by $\zeta\br{M}\epsilon$. This proves claim $1$ of the theorem. Further, since at each step we propagate all interval regions consistent with at least one interval value of the child, the optimal solution to the original problem is feasible for the interval relaxation. Thus, for each we have that $c_i\br{\injp_i^\ast,\injq_i^\ast}\leq c_i\br{\injp_i^{OPT},\injq_i^{OPT}}$. Adding this over all $i$ gives us claim $2$ of the theorem.

Finally, we show that the DP update can be implemented with $\frac{\zeta\br{M}}{\epsilon^5}$ calls to the $\mathrm{PropBound}$ routine. In the loops in algorithm \ref{Alg:DPMain}, we are implicitly looping over possible interval values of $\vs_i,\fpp_i,\fpq_i,\vs_j,\fpp_j,\fpq_j,\vs_k,\fpp_k,\fpq_k$. However, these variables are linked by the constraints:
\begin{align*}
\fpp_i &=\injp_i+\br{\fpp_k-\zr_k	\frac{\fpp_k^2+\fpq_k^2}{\vs_k}}+\br{\fpp_j-\zr_j	\frac{\fpp_j^2+\fpq_j^2}{\vs_j}}	\\
\fpq_i & =\injq_i+\br{\fpq_k-\zi_k	\frac{\fpp_k^2+\fpq_k^2}{\vs_k}}+\br{\fpq_j-\zi_j	\frac{\fpp_j^2+\fpq_j^2}{\vs_j}} \\
\vs_i & =\vs_k+\br{\zr_k^2+\zi_k^2}\frac{\br{\fpp_k^2+\fpq_k^2}}{\vs_k}-2\br{\fpp_k\zr_k+\fpq_k\zi_k} \\
\vs_i & =\vs_j+\br{\zr_j^2+\zi_j^2}\frac{\br{\fpp_j^2+\fpq_j^2}}{\vs_j}-2\br{\fpp_j\zr_j+\fpq_j\zi_j} 
\end{align*}
Thus, if $\injp_i,\injq_i$ are fixed, the $9$ variables are constrained to lie on a $5$-dimensional manifold (since there are 4 non-redundant . This suggests that we only need to do an exhaustive search over a $5$-dimensional space rather than a $9$ dimensional space. 

Suppose $\injp_i\in [\lb{\injp_i}\br{t},\ub{\injp_i}\br{t}],\injq_i\in [\lb{\injq_i}\br{t},\ub{\injq_i}\br{t}]$ and we fix particular interval values (of radius smaller than $\epsilon$) for variables $\vs_i,\fpp_k,\fpq_k,\fpp_j,\fpq_j$. Then, from the first two equations, we know that $\fpp_i,\fpq_i$ must lie in an interval of size $\zeta\br{M}\epsilon$ (since $\injp_i,\injq_i$ lie in intervals of size $\zeta\br{M}$). Thus, we need to loop over at most $\zeta\br{M}$ possible values of $\fpp_i,\fpq_i$. Similarly, if $\vs_i,\fpp_k,\fpq_k$ are fixed to interval values of radius $\epsilon$, the third equation says that $\vs_k$ must lie in an interval of size $\zeta\br{M}\epsilon$, and similarly if $\vs_i,\fpp_j,\fpq_j$ are fixed to intervals of radius of $\epsilon$, $\vs_j$ must lie in an interval of size $\zeta\br{M}\epsilon$. Thus, we need to loop over at most $\zeta\br{M}$ possible values of $\vs_j,\vs_k,\fpp_i,\fpq_i$ once the values of $\vs_i,\fpp_k,\fpq_k,\fpp_j,\fpq_j$ are fixed to intervals of radius $\epsilon$. Finally, the total number of loops in algorithm \ref{Alg:DPUpdate} is at most $\frac{\zeta\br{M}}{\epsilon^5}$. Adding this over all nodes of the network, we get the third claim of Theorem \ref{thm:ApproxMain}.

\end{proof}

\end{document}